\documentclass[aps,prd,reprint,superscriptaddress,amsmath,amssymb,longbibliography,raggedbottom]{revtex4-2}

\usepackage{mathtools}
\usepackage{mathrsfs}
\usepackage{bm}
\usepackage{booktabs}
\usepackage[colorlinks=true,linkcolor=blue,citecolor=blue,urlcolor=blue,bookmarksdepth=2]{hyperref}
\hypersetup{
 pdftitle={Exact thermodynamics and two-horizon relations of four rotating Einstein-Maxwell-Maxwell-dilaton embeddings},
 pdfauthor={Hai Huang, Juhua Chen},
 pdfsubject={Exact thermodynamics of special rotating sectors in four-dimensional Einstein-Maxwell-Maxwell-dilaton theory},
 pdfkeywords={rotating black holes, Einstein-Maxwell-Maxwell-dilaton gravity,
 Kaluza-Klein black holes, horizon identities, black-hole thermodynamics}
}

\newcommand{\dd}{\mathrm d}
\newcommand{\e}{\mathrm e}
\newcommand{\sgn}{\operatorname{sgn}}

\newtheorem{proposition}{Proposition}

\begin{document}

\title{Exact thermodynamics and two-horizon relations of four rotating
Einstein--Maxwell--Maxwell--dilaton embeddings}

\author{Hai Huang}
\email{haihuang@gyu.edu.cn}
\affiliation{College of Mechanical Engineering, Guiyang University,
103 Jianlong Rd., Guiyang, Guizhou 550005, China}

\author{Juhua Chen}
\email{jhchen@hunnu.edu.cn}
\affiliation{Department of Physics, Hunan Normal University,
36 Lushan Rd., Changsha, Hunan 410081, China}

\begin{abstract}
We establish complete outer-horizon thermodynamics for four exact
rotating embeddings of Kerr--Newman and Kaluza--Klein seeds in
four-dimensional, asymptotically flat Einstein--Maxwell--Maxwell--dilaton
theory without a scalar potential. Constant-coefficient field maps
determine the allowed couplings and charge variations. We obtain
single-valued mass functions and verify the first law, Smarr relation,
and scalar-shift identity on the stated regular branches and
finite-area extremal boundaries, including magnetic patch contributions
where required. For two distinct regular horizons with positive finite areas,
the horizon combinations $2T_\pm S_\pm$ and
$S_\pm\Omega_\pm/\pi$ satisfy a common quadratic relation with
the long-range-force parameter determined by the asymptotic mass,
electromagnetic charges, and scalar charge.
\end{abstract}

\maketitle

\section{Introduction}
\label{sec:intro}

Black-hole thermodynamics provides a powerful description of gravitational
systems through horizon quantities and asymptotic conserved charges.
The first law, Smarr relation, and entropy-area relation establish a
connection between horizon geometry and global properties of black holes
\cite{Bekenstein1973,Smarr1973,BardeenCarterHawking1973,Hawking1975,
Wald1993,IyerWald1994,
GibbonsHawking1977,SudarskyWald1992Extrema,
JacobsonKangMyers1994Entropy,Wald2001Thermodynamics,XiongLu2026Interior}. For spherically symmetric dynamical systems,
Hayward's unified first law relates energy change to work and energy
flux \cite{Hayward1998}. When scalar fields and multiple gauge fields are present, however, the
thermodynamic description becomes more involved. Scalar boundary values
can introduce additional work terms in the first law, while the
coupling between scalar hair, gauge charges, and horizon quantities
depends on the allowed field configurations
\cite{GibbonsKalloshKol1996,Astefanesei2018ScalarFirstLaw,Pacilio2018,
LiuLu2014ScalarCharges,LuPangPope2013DyonicAdS,
CardenasFuentealbaMatulich2016,AstefaneseiGoldsteinMahapatra2008,
Li2021ScalarHairThermodynamics,BallesterosGomezFayrenOrtinZatti2023,BechmannLechtenfeld1995,AnabalonOliva2012,
HerdeiroRadu2014ScalarHair,HerdeiroRadu2015Construction,HerdeiroRadu2015}.

For rotating dilaton black holes, explicit solutions are available only
in special coupling sectors \cite{HorneHorowitz1992,KunzMaisonNavarroViebahn2006}, while generic rotating solutions with
multiple gauge fields and nontrivial scalar profiles remain difficult to
construct. This motivates the question of whether thermodynamic
quantities and relations can be obtained without first constructing the
full metric and matter fields. Recent works have shown that thermodynamic quantities and relations
can be reconstructed without first constructing explicit black hole
solutions, first for EMMD black holes and subsequently for more general
EMD-like theories
\cite{YangLuLu2026PRL,WuLu2026GeneralEMD}. The method combines the first law, scaling
relations, connections between asymptotic and horizon quantities, and a
secondary (non-independent) scalar-charge relation supported by regular
static solutions
\cite{LuYang2013Toda,LuYangLu2025MassCharge}. Exact one-field and Toda
solutions provide independent checks of this framework
\cite{GibbonsMaeda1988,GarfinkleHorowitzStrominger1991,
GibbonsWiltshire1986,LuYang2013Toda,LuWuWuZhao2026Toda}.

This development motivates extending solution-independent
thermodynamic methods to rotating systems, where explicit metric
constructions are considerably more difficult. For rotating EMMD systems,
an important question is how such thermodynamic structures extend to
theories with multiple Maxwell fields through field maps, constrained
charge spaces, and relations between asymptotic charges and horizon
geometry.

Rotation introduces an additional dimensionless parameter that is not
fixed by the static thermodynamic relations. At fixed couplings and
scalar boundary value, asymptotically flat scaling gives
\begin{equation*}
 M(\lambda^2S,\lambda^2J,\lambda Q_i,\lambda P_i;\phi_0)
 =\lambda M(S,J,Q_i,P_i;\phi_0).
\end{equation*}
Both $S$ and $J$ scale as length squared, so $J/S$ remains invariant
under this scaling. The static mass relation determines only the
$J=0$ sector, and invariance under spin reversal implies that the mass
is an even function of $J$, but it does not determine the dependence
on the independent ratio $J/S$. Therefore, additional rotational
information is required. The horizon area, surface gravity, and
electromagnetic potentials all acquire spin dependence.

To determine this rotational information, we use known exact
rotating one-field solutions as seeds and embed them into the
two-Maxwell EMMD theory. The relevant seeds are Kerr--Newman
\cite{NewmanEtAl1965,Carter1968,ChristodoulouRuffini1971}
and the electric and dyonic Kaluza--Klein black holes at $|a|=\sqrt3$
\cite{Rasheed1995,Larsen2000}. These seeds provide exact rotating
geometries with known thermodynamic properties and allow the separation
of geometric information from the Maxwell charge realization.
A useful embedding must answer two questions. After one Maxwell field
is replaced by two fields, do the metric and scalar still solve the
EMMD equations? How are the seed charges distributed between the two
target fields? These questions determine the allowed couplings,
independent charge variations, and electromagnetic work. Magnetic
charges also require regular north and south gauge patches.

Several ingredients are already known for the seeds. Chow and Comp\`ere
\cite{ChowCompere2014} give the quadratic relation between asymptotic
charges and the outer-horizon temperature--entropy and rotational
products, together with its inner-horizon expression. Cveti\v c et al.
\cite{CveticEtAl2018} give the temperature--entropy sign relation and
the branch-dependent sign of the rotational product; they also give the
single-charge rotating KK mass formula and its auxiliary cubic equation.
Related charged rotating solutions and near-horizon structures are
known in higher-dimensional string compactifications
\cite{CveticYoum1996FiveDimensional,CveticLarsen1998}.
Here we use the four-dimensional seed thermodynamics to determine
its realization in the two-field EMMD theory.

The first result concerns four such realizations, listed in
Table~\ref{tab:sector-atlas}: equal-coupling Kaluza--Klein, zero-coupling
Kerr--Newman, balanced Kerr--Newman, and Rasheed--Larsen direct--dual
embeddings. The classification concerns real constant-coefficient
combinations of a seed field and its scalar-weighted dual that preserve
the seed metric and the full scalar profile. The resulting field maps, charge spaces,
and work forms give complete outer-horizon thermodynamics on the allowed
states. We establish single-valued mass functions, identify their physical
roots and regular boundaries, and verify the first law, Smarr relation,
and scalar-shift identity for every family. These EMMD embedding and
state-space results build on the known geometries and seed thermodynamics.
Independent conserved charges can share a common Maxwell profile;
charge-space dimension and field-strength rank therefore describe
different freedoms.

For these families, ``all-spin'' means the full allowed range of the
single angular momentum on the specified connected outer-horizon branch.
``Global'' means the complete regular outer-horizon thermodynamic domain
inherited from each seed, including the stated finite-area extremal
boundaries. 

The second result is a common outer- and inner-horizon relation on the
subsets with two distinct regular horizons of positive finite area.
We extend the long-range-force parameter introduced in
Ref.~\cite{YangLuLu2026PRL} to arbitrary scalar boundary values and to
configurations carrying both electric and magnetic charges, thereby
obtaining the scalar-shift-covariant quantity $\mu_{\rm LR}$.
\begin{equation}
 \mu_{\rm LR}^{2}
 =
 (2TS)^2+\left(\frac{S\Omega_H}{\pi}\right)^2 .
\label{eq:intro-unified-rotating-identity}
\end{equation}
For each of the four asymptotically flat families, we evaluate
$\mu_{\rm LR}$ from the asymptotic mass, scalar charge, and
electromagnetic charges in the EMMD normalization.
Independently, we compute the horizon entropies, signed temperatures,
and angular velocities from the geometry.
Comparing these asymptotic and horizon quantities, we find that both
horizons satisfy the same relation.

The paper is organized as follows. Section~\ref{sec:framework} gives the theory, charges, and thermodynamic
conventions. Section~\ref{sec:seed-classification} classifies the seed
embeddings. Sections~\ref{sec:o2-kk}--\ref{sec:hodge} develop the four
families and their complete outer-horizon thermodynamics.
Section~\ref{sec:two-horizon} states the common horizon relations, checks
them for each family, and explains their domains, limits, and physical
meaning. Section~\ref{sec:conclusions} concludes. The appendices give the
gauge-patch and horizon-regularity calculations and the charge-space and
global mass-inversion proofs.

We use signature $(-+++)$, orientation
$\dd t\wedge\dd r\wedge\dd\theta\wedge\dd\varphi>0$, and $G_4=1$.

\begin{table*}[t]
 \caption{
 Couplings, field maps, seed solutions, and outer-horizon domains
 for the four EMMD embeddings.
 At fixed $\phi_0$, $d_Q$ counts the independent charge parameters,
 and $r_F$ is the generic number of linearly independent
 field-strength profiles in the electromagnetic frame used for
 each row.
 The rows distinguish EMMD embeddings, not necessarily different
 spacetime geometries.
 }
 \label{tab:sector-atlas}

 \centering\small
 \renewcommand{\arraystretch}{1.15}
 \setlength{\tabcolsep}{2.5pt}
 \newcommand{\atlascell}[2]{%
   \parbox[t]{#1}{\raggedright #2\par}%
 }

 \begin{tabular}{@{}llllll@{}}
 \toprule

 \atlascell{0.145\textwidth}{Family}
 &
 \atlascell{0.120\textwidth}{$(a_1,a_2)$}
 &
 \atlascell{0.215\textwidth}{Field map; scalar}
 &
 \atlascell{0.150\textwidth}{Known seed}
 &
 \atlascell{0.085\textwidth}{$(d_Q,r_F)$}
 &
 \atlascell{0.220\textwidth}{Outer-horizon domain}
 \\

 \midrule

 \atlascell{0.145\textwidth}{Equal-coupling  Kaluza–Klein}
 &
 \atlascell{0.120\textwidth}{$(-\sqrt3,-\sqrt3)$}
 &
 \atlascell{0.215\textwidth}{
   Internal split of one field;
   nonconstant scalar on the charged branch
 }
 &
 \atlascell{0.150\textwidth}{
   Boosted-Kerr KK seed
   \cite{Rasheed1995,Larsen2000}
 }
 &
 \atlascell{0.085\textwidth}{$(2,1)$}
 &
 \atlascell{0.220\textwidth}{
   Nonextremal outer branch and finite-area extremal limits
 }
 \\

 \addlinespace[3pt]

 \atlascell{0.145\textwidth}{Zero-coupling Kerr–Newman}
 &
 \atlascell{0.120\textwidth}{$(0,0)$}
 &
 \atlascell{0.215\textwidth}{
   Internal split of one field;
   constant scalar
 }
 &
 \atlascell{0.150\textwidth}{
   Kerr--Newman
   \cite{ChristodoulouRuffini1971}
 }
 &
 \atlascell{0.085\textwidth}{$(2,1)$}
 &
 \atlascell{0.220\textwidth}{
   Purely electric outer branch and finite-area extremal limits
 }
 \\

 \addlinespace[3pt]

 \atlascell{0.145\textwidth}{Balanced Kerr–Newman}
 &
 \atlascell{0.120\textwidth}{Any $a_1>0>a_2$}
 &
 \atlascell{0.215\textwidth}{
   Parallel fields with cancelling scalar sources;
   constant scalar
 }
 &
 \atlascell{0.150\textwidth}{
   Dyonic Kerr--Newman seed
   \cite{LuYangLu2025MassCharge}
 }
 &
 \atlascell{0.085\textwidth}{$(2,1)$}
 &
 \atlascell{0.220\textwidth}{
   Balanced outer branch and finite-area extremal limits,
   including static extremal RN
 }
 \\

 \addlinespace[3pt]

 \atlascell{0.145\textwidth}{Rasheed--Larsen direct--dual}
 &
 \atlascell{0.120\textwidth}{$(\sqrt3,-\sqrt3)$}
 &
 \atlascell{0.215\textwidth}{
   Seed field and its scalar-weighted dual;
   generally nonconstant scalar
 }
 &
 \atlascell{0.150\textwidth}{
   Dyonic Rasheed--Larsen seed
   \cite{Rasheed1995,Larsen2000}
 }
 &
 \atlascell{0.085\textwidth}{$(3,2)$}
 &
 \atlascell{0.220\textwidth}{
   Nonextremal outer branch and the fast and slow
   finite-area extremal limits
 }
 \\

 \bottomrule
 \end{tabular}
\end{table*}

\section{Theory and thermodynamic conventions}
\label{sec:framework}

\subsection{Action, field equations}

We consider four-dimensional Einstein--Maxwell--Maxwell--dilaton gravity
with two Maxwell fields and a dilaton $\phi$. Throughout this paper the
scalar potential vanishes, and the action is

\begin{equation}
 \begin{aligned}
 I={}&\frac1{16\pi}\int \dd^4x\sqrt{-g}\left[
 R-\frac12(\partial\phi)^2\right.\\
 &\left.\hspace{8mm}-\frac14\sum_{i=1}^2\e^{a_i\phi}F_i^2\right],\\
 F_i={}&\dd A_i.
 \end{aligned}
 \label{eq:general-action}
\end{equation}

The dimensionless couplings $a_i$ are fixed within each theory, and
$F_i=\dd A_i$ holds in each gauge patch. The field equations are
\begin{equation}
\begin{aligned}
 \dd F_i={}&0,\quad \dd(\e^{a_i\phi}\star F_i)=0,\quad
 \Box\phi={}\frac14\sum_i a_i\e^{a_i\phi}F_i^2,\\
 R_{\mu\nu}={}&\frac12\partial_\mu\phi\partial_\nu\phi 
 +\frac12\sum_i\e^{a_i\phi}\left(
 F_{i\mu\rho}F_{i\nu}{}^\rho-\frac14g_{\mu\nu}F_i^2\right).
\end{aligned}
 \label{eq:general-scalar}
\end{equation}
The absence of a scalar potential leaves the dilaton couplings to the
Maxwell fields intact, so the scalar can remain nonconstant.
We impose asymptotically flat boundary conditions and consider stationary,
axisymmetric black holes with one connected regular outer horizon.
All matter fields inherit the Killing symmetries up to gauge transformations.
The time translation is normalized at infinity and the azimuthal period
is $2\pi$. These choices fix the temperature and angular velocity used
in the thermodynamic and two-horizon checks.

\subsection{Conserved charges and conjugate potentials}

We normalize the asymptotic scalar and electromagnetic charges as
\begin{equation}
 \begin{aligned}
 \phi&=\phi_0+\frac{2\Sigma}{r}+O(r^{-2}),\\
 P_i&=\frac1{16\pi}\int_{S^2_\infty}F_i,\qquad
 Q_i=\frac1{16\pi}\int_{S^2_\infty}\e^{a_i\phi}\star F_i.
 \end{aligned}
 \label{eq:general-charges}
\end{equation}
Let $\chi=\partial_t+\Omega_H\partial_\varphi$. In the absence of magnetic monopole charge, a gauge potential that vanishes at infinity gives $\Phi_i=(\chi\mathbin{\cdot}A_i)|_\infty-(\chi\mathbin{\cdot}A_i)|_H$. When magnetic charge is present, regular potentials must be defined on overlapping north and south patches. Their horizon contractions can differ by gauge terms. We use the electric and magnetic conjugates of the first law defined by covariant phase space, with a consistent asymptotic reference and the required gauge and patch-boundary contributions \cite{Prabhu2017FirstLaw,ElgoodMeessenOrtin2020FirstLaw,OrtinPereniguez2022Magnetic}. The dual potentials satisfy
\begin{equation}
 \dd\widetilde A_i=-\e^{a_i\phi}\star F_i.
 \label{eq:dual-potentials}
\end{equation}
Appendix~\ref{app:gauge-patches} gives the regular patch potentials and their horizon limit.

Define the long-range-force parameter by
\begin{equation}
 \mu_{\rm LR}^2:=M^2+\Sigma^2
 -4\sum_i\left(\e^{-a_i\phi_0}Q_i^2+\e^{a_i\phi_0}P_i^2\right).
 \label{eq:mu-LR-general}
\end{equation}
This extends the static parameter $\mu$ of Ref.~\cite{YangLuLu2026PRL} to arbitrary scalar boundary values and to both electric and magnetic charges. The exponential factors make the definition invariant under the constant scalar shift. We take $\mu_{\rm LR}$ to be the nonnegative square root when the right-hand side is nonnegative, as it is for the three nonzero-coupling families and the zero-coupling control below.

\subsection{First law and variations on constrained state spaces}
\label{subsec:intrinsic-derivatives}

On a connected asymptotically flat branch with a differentiable mass function, the extended first law and four-dimensional Smarr relation are
\begin{align}
 &\dd M={}T\dd S+\Omega_H\dd J
 +\sum_i\left(\Phi_i\dd Q_i+\Psi_i\dd P_i\right)
 -\frac12\Sigma\dd\phi_0, \notag\\
 &M={}2TS+2\Omega_HJ
 +\sum_i\left(\Phi_iQ_i+\Psi_iP_i\right).
 \label{eq:general-Smarr}
\end{align}
Under the constant shift $\phi\mapsto\phi+c$, $A_i\mapsto\e^{-a_ic/2}A_i$, the electric charges scale with weight $a_i/2$ and the magnetic charges with weight $-a_i/2$; their potentials have the opposite weights. The metric, mass, entropy, and angular momentum stay the same. Substituting the changes in charges and scalar boundary value into the first law gives the scalar-shift Ward identity
\begin{equation}
 \Sigma=\sum_i a_i\left(Q_i\Phi_i-P_i\Psi_i\right).
 \label{eq:general-Ward}
\end{equation}
Setting $P_i=0$ gives its purely electric form.

The charges and scalar boundary value of an embedded family may be constrained and therefore cannot always be varied independently. Although the mass is written as \(M(S,J,Q_i,P_i;\phi_0)\), this notation does not imply that all these variables can be varied independently within a particular solution family. The first law must be verified for variations allowed by that family. Different extensions of the mass can agree throughout the family while giving different partial derivatives in directions that leave it. The first law therefore determines the total work associated with allowed variations, whereas physical responses in these additional directions require a larger solution family that admits the corresponding variations.

Although the constrained first law fixes only the total work term for allowed variations, the individual electric and magnetic potentials can be determined from the actual gauge fields, the chosen reference conventions at infinity, and the magnetic gauge-patch contributions. These additional inputs fix the specific potentials used here. Rather than simply taking partial derivatives of a formal mass expression, we first determine the potentials from the fields and their gauge patches, and then verify that the resulting work terms reproduce the first law for the allowed charge variations.

\section{Classification of constant-coefficient seed embeddings}
\label{sec:seed-classification}

We classify a specific way of embedding one-field solutions into
EMMD theory: we keep the seed metric and scalar field unchanged
and construct the two Maxwell fields from the seed field and its
scalar-weighted dual, using real coefficients that are constant
throughout spacetime. Both the seed and target theories have
zero scalar potential.

For constant-scalar seeds, the completeness result assumes that
$F^2$ and $F\cdot\star F$ are linearly independent as functions
of position. This means that no nontrivial combination of these
two functions with constant real coefficients vanishes everywhere.
Under this assumption, the coefficients of the two invariants
in the scalar equation must vanish separately. The resulting
source-cancelling maps also work when the invariants are dependent
or vanish, but such special seeds may admit additional maps.

Here a constant scalar means $\phi(x)\equiv\phi_0$ throughout
the spacetime region under consideration. A vanishing coefficient
of the asymptotic $1/r$ term alone is not sufficient.
The results below therefore identify all field maps of the stated
form under their respective assumptions, rather than all rotating
EMMD solutions.

\subsection{Nonconstant-scalar seeds}
\label{subsec:nonconstant-scalar-embeddings}

We first ask which constant combinations of a seed Maxwell field and its scalar-weighted dual preserve the seed metric and dilaton while satisfying the two-field equations.

\begin{proposition}[Constant-coefficient maps]
\label{prop:constant-coefficient-seed-dual}
Let $a,a_i\in\mathbb R$, with $a\ne0$, and let $(g,\phi,H)$ solve the one-field version of Eq.~\eqref{eq:general-action} with coupling $a$ on a connected open set $U$. Assume that $H$ is nonzero at every point and that $\phi(U)$ contains an interval of nonzero length. In a fixed electric frame, take real constant coefficients $u_i,v_i$ and set
\begin{equation}
 F_i=u_iH+v_i\e^{a\phi}\star H,
 \qquad i=1,2.
 \label{eq:general-seed-dual-ansatz}
\end{equation}
Then $(g,\phi,F_1,F_2)$ satisfies the two-field equations written in terms of field strengths on $U$ if and only if
\begin{equation}
 u_i(a_i-a)=0,\quad v_i(a_i+a)=0,\quad
 \sum_i(u_i^2+v_i^2)=1.
 \label{eq:general-seed-dual-classification}
\end{equation}
If $(u_i,v_i)=(0,0)$, the corresponding field vanishes and its coupling is unrestricted.
\end{proposition}

\noindent\textit{Proof.}
For a two-form $K$, define its Maxwell stress tensor by
\[
 {\cal T}_{\mu\nu}[K]
 =K_{\mu\rho}K_\nu{}^\rho
 -\frac14g_{\mu\nu}K_{\rho\sigma}K^{\rho\sigma}.
\]
In four-dimensional Lorentzian signature, $\star^2=-1$ on real two-forms, and
\[
 {\cal T}[\alpha H+\beta\star H]
 =(\alpha^2+\beta^2){\cal T}[H]
\]
for real $\alpha$ and $\beta$. The target Maxwell contribution to the Einstein equation is therefore
\begin{equation}
\begin{aligned}
 &\sum_i\e^{a_i\phi}{\cal T}[F_i]
 =\e^{a\phi}C(\phi){\cal T}[H], \\
 &C(\phi):=\sum_i u_i^2\e^{(a_i-a)\phi}
 +\sum_i v_i^2\e^{(a_i+a)\phi}.
 \label{eq:general-seed-dual-source}
\end{aligned}
\end{equation}
A nonzero real Maxwell field has positive energy density, proportional to $|\mathbf E|^2+|\mathbf B|^2$ in a local orthonormal frame. Its stress tensor is therefore nonzero. Matching the seed and target Einstein equations requires $C(\phi)=1$ throughout $\phi(U)$. Differentiating twice gives
\begin{equation}
\begin{aligned}
 0=C''(\phi)
 ={}&\sum_i u_i^2(a_i-a)^2\e^{(a_i-a)\phi}\\
 +&\sum_i v_i^2(a_i+a)^2\e^{(a_i+a)\phi}.
 \label{eq:general-seed-dual-convexity}
\end{aligned}
\end{equation}
Each term is nonnegative. Hence $u_i(a_i-a)=0$ and $v_i(a_i+a)=0$. With these conditions, $C(\phi)=1$ becomes the normalization condition in Eq.~\eqref{eq:general-seed-dual-classification}, proving necessity.

Conversely, the seed equations $\dd H=0=\dd(\e^{a\phi}\star H)$ imply $\dd F_i=0$, and
\begin{equation}
\begin{aligned}
 \dd(\e^{a_i\phi}\star F_i)
 ={}&u_i(a_i-a)\e^{a_i\phi}
       \dd\phi\wedge\star H\\
 -&v_i(a_i+a)\e^{(a_i+a)\phi}
       \dd\phi\wedge H,
 \label{eq:general-seed-dual-Maxwell}
\end{aligned}
\end{equation}
vanishes by Eq.~\eqref{eq:general-seed-dual-classification}. The Einstein sources agree by Eq.~\eqref{eq:general-seed-dual-source}. Since $a\ne0$, the coupling conditions also give $u_iv_i=0$ for each field. Using $(\star H)^2=-H^2$, the target scalar source is
\begin{equation}
 \sum_i a_i\e^{a_i\phi}F_i^2
 =a\e^{a\phi}\left[\sum_i(u_i^2+v_i^2)\right]H^2
 =a\e^{a\phi}H^2,
 \label{eq:general-seed-dual-scalar-match}
\end{equation}
which equals the seed scalar source. This proves sufficiency.

When both target fields are nonzero,
Eq.~\eqref{eq:general-seed-dual-classification} requires
\begin{equation}
 (a_1,a_2)\in
 \{(a,a),\,(-a,-a),\,(a,-a),\,(-a,a)\}.
 \label{eq:allowed-coupling-representatives}
\end{equation}
For $(a_1,a_2)=(a,a)$, both target fields are constant
multiples of the seed field $H$. For
$(a_1,a_2)=(-a,-a)$, both are constant multiples of its
scalar-weighted dual $e^{a\phi}\star H$. For the two
opposite-sign coupling pairs, one target field is
proportional to $H$ and the other to $e^{a\phi}\star H$,
with their roles exchanged between the two cases.

When the scalar is nonconstant and both target fields are nonzero, this method requires the two coupling constants to have equal magnitudes. Obtaining solutions with a nonconstant scalar, two nonzero fields, and unequal coupling magnitudes therefore requires going beyond this construction. For each seed, we must separately check that the gauge potentials are consistently defined on overlapping patches and that the solution is regular in the specified region.

\subsection{Constant-scalar seeds}
\label{subsec:constant-scalar-embeddings}

We next ask how the independent electromagnetic invariants of a constant-scalar seed constrain the two-field embedding.

\begin{proposition}[Constant-scalar maps]
\label{prop:constant-scalar-seed-dual}
Let $(g,F)$ be a nonzero real source-free Einstein--Maxwell solution on a connected open set $U$. Take $a_1a_2\ne0$ and $\phi=\phi_0$ constant. Assume that $F^2$ and $F\mathbin{\cdot}\star F$ are linearly independent over $\mathbb R$ as functions on $U$. For real constant coefficients, set
\begin{equation}
 F_i=\e^{-a_i\phi_0/2}(u_iF+v_i\star F),
 \qquad z_i=u_i+\mathrm i v_i.
 \label{eq:constant-scalar-ansatz}
\end{equation}
The constant-scalar EMMD equations hold on $U$ if and only if
\begin{equation}
 \sum_{i=1}^2|z_i|^2=1,
 \qquad
 \sum_{i=1}^2a_i z_i^2=0.
 \label{eq:constant-scalar-classification}
\end{equation}
For $a_1a_2<0$, all solutions have the form
\begin{equation}
 z_i=\epsilon_i\sqrt{w_i}\,\e^{\mathrm i\vartheta_0},
 \quad
 (w_1,w_2)=\frac{1}{a_1-a_2}(-a_2,a_1).
 \label{eq:constant-scalar-opposite}
\end{equation}
For $a_1a_2>0$, they have the form
\begin{equation}
 \begin{aligned}
 z_1&=\epsilon_1\sqrt{\widetilde w_1}\,\e^{\mathrm i\vartheta_0}, \quad \epsilon_i=\pm1,\\
 z_2&=\mathrm i\epsilon_2\sqrt{\widetilde w_2}\,\e^{\mathrm i\vartheta_0}, \quad
 (\widetilde w_1,\widetilde w_2)=\frac{1}{a_1+a_2}(a_2,a_1).
 \end{aligned}
 \label{eq:constant-scalar-same}
\end{equation}
Thus a common duality rotation and independent field signs reduce the two cases to the parallel balanced form and the direct--Hodge form.
\end{proposition}

\noindent\textit{Proof.}
The seed Maxwell equation and Bianchi identity give the corresponding target equations. The required two-form identities are
\begin{align}
 {\cal T}[uF+v\star F]&=(u^2+v^2){\cal T}[F],\notag\\
 (uF+v\star F)^2&=(u^2-v^2)F^2
 +2uvF\mathbin{\cdot}\star F.
 \label{eq:constant-scalar-source-algebra}
\end{align}
The seed stress tensor is nonzero wherever $F\ne0$, so Einstein-source matching gives the first condition in Eq.~\eqref{eq:constant-scalar-classification}. The scalar equation becomes
\begin{equation}
 \left[\sum_i a_i(u_i^2-v_i^2)\right]F^2
 +2\left[\sum_i a_i u_iv_i\right]
 F\mathbin{\cdot}\star F=0.
 \label{eq:constant-scalar-source}
\end{equation}
Independence of the two invariants requires both coefficients to vanish, which gives the second condition in Eq.~\eqref{eq:constant-scalar-classification}. Solving $z_2^2=-(a_1/a_2)z_1^2$ and imposing the normalization gives Eqs.~\eqref{eq:constant-scalar-opposite} and \eqref{eq:constant-scalar-same}. \hfill$\square$

The common angle \(\vartheta_0\) controls the mixing of the seed electric and magnetic charges. Changing this angle leaves the seed electromagnetic stress tensor unchanged, so the original metric and the transformed field still satisfy the Einstein–Maxwell equations. For some special seeds, the two electromagnetic invariants are proportional or vanish. This reduces the constraints imposed by the scalar equation and may allow additional field combinations. However, when both target fields are constant multiples of the same nonzero field \(F\), the weights for opposite-sign couplings remain uniquely determined as long as \(F^2\not\equiv0\).

\subsection{Charge freedom and the resulting families}
\label{subsec:embedding-family-overview}

\label{sec:embedding-diagnostics}
Let $\mathcal C_{\phi_0}$ be the allowed charge set at fixed scalar boundary value. On each smooth part of this set, define
\begin{equation}
 d_Q:=\dim\mathcal C_{\phi_0},
 \qquad
 r_F:=\dim_{\mathbb R}\operatorname{span}\{F_1,F_2\}.
 \label{eq:rank-definitions}
\end{equation}
The quantities \(d_Q\) and \(r_F\) count the independent charge parameters and the linearly independent field-strength profiles, respectively. The former describes how the charges can vary between black-hole states, while the latter indicates whether the two fields in a given state are simply constant multiples of one another. Thus, two independent charges can share a common field profile, while two linearly independent field profiles can both be constructed from the same seed.

Table~\ref{tab:sector-atlas} lists the four thermodynamic realizations. Different field maps can describe the same seed geometry. An internal split $F_i=n_iF$, $\sum_i n_i^2=1$, gives the equal-coupling Kaluza--Klein family and the zero-coupling Kerr--Newman example. Cancellation of the scalar source gives the balanced family for $a_1a_2<0$. At $a_1=-a_2=\sqrt3$, pairing a Rasheed--Larsen field with its dual gives a magnetic family with generic $r_F=2$ in the original electric frame. 

The equal-coupling and zero-coupling branches have two independent charges, although both fields are constant multiples of the same field profile. The balanced branch also has one common field profile, with two independent parameters specifying the electric and magnetic charge amplitudes. Restricting this branch to purely electric configurations leaves only one charge parameter. Applying an electromagnetic duality transformation to the second field of the balanced branch makes the two nonzero fields no longer proportional, but does not increase the number of independent charge parameters or produce a new seed geometry. The Rasheed–Larsen branch has three independent charge parameters. Its four electromagnetic charges satisfy one constraint, and the first law must be verified for variations that preserve this constraint. The same-sign constant-scalar representative in Sec.~\ref{cor:direct-hodge-constant-scalar} has $(d_Q,r_F)=(2,2)$ and is related to the balanced representative by dualizing the second Maxwell field, with the corresponding change of coupling and charges.

\section{Equal-coupling Kaluza--Klein thermodynamics}
\label{sec:o2-kk}

\subsection{Embedding and charge variations}

Take
\begin{equation}
 a_1=a_2=-\sqrt3.
 \label{eq:o2-couplings}
\end{equation}
The seed is the rotating Kaluza--Klein black hole obtained by reducing
five-dimensional vacuum gravity \cite{Rasheed1995,Larsen2000}. With our
conventions, the reduction
\begin{equation}
 \dd\widehat s_5^2=\e^{\phi/\sqrt3}\dd s_4^2
 +\e^{-2\phi/\sqrt3}(\dd z+{\cal A})^2,
 \label{eq:o2-KK-ansatz}
\end{equation}
maps a boosted Kerr string to a four-dimensional rotating electric solution
with field strength ${\cal F}=\dd{\cal A}$ and a varying dilaton.
For a constant unit vector $n=(n_1,n_2)=(\cos\psi,\sin\psi)$, define
\begin{equation}
 F_i=n_i{\cal F},\qquad n_1^2+n_2^2=1.
 \label{eq:o2-split}
\end{equation}
The two Maxwell equations are constant multiples of the seed equation, and
\begin{align}
 \sum_i\e^{-\sqrt3\phi}F_i^2
 &=\e^{-\sqrt3\phi}{\cal F}^2,\notag\\
 \sum_i\e^{-\sqrt3\phi}{\cal T}_{\mu\nu}[F_i]
 &=\e^{-\sqrt3\phi}{\cal T}_{\mu\nu}[{\cal F}],
 \label{eq:o2-source-check}
\end{align}
where ${\cal T}_{\mu\nu}[K]=K_{\mu\rho}K_\nu{}^\rho-g_{\mu\nu}K^2/4$.
This is the equal-coupling internal split of
Proposition~\ref{prop:constant-coefficient-seed-dual}, valid throughout the
seed spin range. The one-charge rotating Kaluza--Klein
Christodoulou--Ruffini relation and its cubic in $\cosh\delta$ are given in
Ref.~\cite{CveticEtAl2018}.
We express them in the two-charge normalization, include the internal
rotation, and determine the unique physical root over the full outer branch.

\subsection{Complete outer-horizon thermodynamic quantities}
\label{subsec:o2-outer-thermodynamics}

Let $\delta$ be the boost parameter, and define
\begin{equation}
 \begin{aligned}
 \ell&=\sqrt{m^2-b^2},& r_\pm&=m\pm\ell.
 \end{aligned}
 \label{eq:o2-parameters}
\end{equation}
The parameters \(m\) and \(b\) are the mass and rotation parameters of the original Kerr seed. A boost along the fifth dimension generates the electric charge, and \(\delta\) measures the strength of this boost. The resulting physical mass is \(M=m(1+\tfrac12\sinh^2\delta)\), so \(m\) generally differs from \(M\). The parameters \(m,\delta\) determine the total charge amplitude, while the unit vector \(n\) determines its distribution between the two fields.
The charged nonextremal outer branch has $m>0$, $|b|<m$, $\delta>0$,
and $n\in S^1$. Fixing the signs of nonzero $b,n_1,n_2$ selects a connected
rotating two-charge component. At $\phi_0=0$, the thermodynamic quantities are

\begin{equation}
\begin{aligned}
 M&=m\left(1+\frac{\sinh^2\delta}{2}\right),
 &J&=mb\cosh\delta,\\
 Q_i&=\frac12n_i m\sinh\delta\cosh\delta,
 &\Sigma&=-\frac{\sqrt3}{2}m\sinh^2\delta,\\
 P_i&=0,
 &S&=2\pi m\cosh\delta(m+\ell),\\
 \Phi_i&=n_i\tanh\delta,
&T&=\frac{\ell}{4\pi m\cosh\delta(m+\ell)},\\
\Omega_H&=\frac{b}{2m\cosh\delta(m+\ell)},
 &A&_H=4S,\quad \quad \kappa=2\pi T.
\end{aligned}
\label{eq:o2-area-kappa}
\end{equation}

At $\phi_0=0$, these formulas give
\begin{equation}
 \begin{aligned}
 \mu_{\rm LR}^2
 &=M^2+\Sigma^2-4(Q_1^2+Q_2^2)=m^2,\\
 2TS&=\ell=\sqrt{m^2-b^2}.
 \end{aligned}
\label{eq:o2-static-rotation-bridge}
\end{equation}
Applying the scalar shift restores an arbitrary boundary value,
\begin{equation}
 M^2+\Sigma^2
 -4\e^{\sqrt3\phi_0}(Q_1^2+Q_2^2)=m^2.
 \label{eq:o2-shift-covariant-mu}
\end{equation}

Together with $S\Omega_H/\pi=b$, this proves
Eq.~\eqref{eq:intro-unified-rotating-identity} at arbitrary $\phi_0$.
At $b=0$, $\mu_{\rm LR}=2TS$; at finite spin the difference of their squares
is $b^2$.

\subsection{First law, Smarr, and scalar shift}
\label{subsec:o2-direct-checks}

We first check the parameter variations at fixed $\phi_0=0$.
Differentiating the thermodynamic quantities, including their
dependence on the internal angle $\psi$, gives

\begin{align} 
 &M
 =2TS+2\Omega_HJ+\sum_i\Phi_iQ_i,
 &
 \Sigma
 &=-\sqrt3\sum_iQ_i\Phi_i.\notag\\
 &\dd M
 =T\dd S+\Omega_H\dd J+\sum_i\Phi_i\dd Q_i,
 \label{eq:o2-Smarr-Ward}
\end{align}

For an explicit check, use $\sum_i n_i\dd n_i=0$.
The geometric and electric contributions are
\begin{align}
  &\sum_i\Phi_i\dd Q_i
=\frac{\sinh^2\delta}{2}\dd m
   +\frac{m(2\cosh^2\delta-1)}{2}
    \tanh\delta\,\dd\delta,\notag\\
 &T\dd S+\Omega_H\dd J
 =\dd m+\frac{m}{2}\tanh\delta\,\dd\delta.
\end{align}
Their sum equals
$\dd[m(\cosh^2\delta+1)/2]=\dd M$;
all $\dd b$ and $\dd\psi$ terms cancel. Moreover,
\begin{equation}
 \begin{aligned}
 &2TS+2\Omega_HJ=m,\\
 &\sum_i\Phi_iQ_i
 =\frac{m\sinh^2\delta}{2}
 =-\frac{\Sigma}{\sqrt3},
 \end{aligned}
\end{equation}
which verifies both the Smarr relation and the scalar shift identity.

The two charges are determined by an overall magnitude and a distribution angle. Changing only the angle changes the charge carried by each field while keeping \(Q_1^2+Q_2^2\) fixed, so the two electromagnetic work terms cancel.
The charges are still independent coordinates: with
$R=\sqrt{Q_1^2+Q_2^2}$,
\begin{equation}
 \begin{aligned}
 &\left|
 \det\frac{\partial(M,J,Q_1,Q_2)}
 {\partial(m,b,\delta,\psi)}
 \right|\\
 &\qquad=
 \frac{Rm^2\cosh\delta}{2}
 \left(1+\frac32\sinh^2\delta\right)>0,
 \end{aligned}
 \label{eq:o2-Jacobian}
\end{equation}
so $(Q_1,Q_2)$ provide two thermodynamic coordinates before
identifying states related by the global $O(2)$ symmetry.
The field rank is one.

\subsection{ Single-valued mass function}

Eliminating the seed parameters does not produce a different thermodynamic description. It rewrites the same outer-horizon thermodynamics in terms of the entropy, angular momentum, charges, and scalar boundary value. The single-valued mass function assigns a mass to each allowed state, and its derivatives along allowed variations must reproduce the temperature, angular velocity, electromagnetic work, and scalar response. For the equal-coupling branch, these derivatives agree with the quantities calculated directly from the outer-horizon geometry and gauge fields. The mass function therefore connects the thermodynamic quantities, while the physical root and regular domain ensure that the relation describes the chosen outer-horizon branch. We define
\begin{equation}
 \widehat Q_i=\e^{\sqrt3\phi_0/2}Q_i \quad
 \widehat R=\sqrt{\widehat Q_1^2+\widehat Q_2^2}
 \quad L=\frac{S}{2\pi}.
 \label{eq:o2-dressed}
\end{equation}
For $L>|J|$ and $\widehat R>0$, determine the boost parameter
$\delta>0$ from
\begin{equation}
 \cosh^3\delta-\cosh\delta
 =\frac{8L\widehat R^2}{J^2+L^2}.
 \label{eq:o2-cubic}
\end{equation}
The left-hand side is strictly increasing for $\delta\ge0$
and ranges from zero to infinity, so this equation determines
a unique $\delta>0$.
The mass and seed parameters are then
\begin{equation}
\begin{aligned}
 &M^2=\frac{(J^2+L^2)(\cosh^2\delta+1)^2}{8L\cosh\delta},\\
 &b=\frac{J}{m\cosh\delta}, \qquad m=\sqrt{\frac{J^2+L^2}{2L\cosh\delta}},\\
 &\ell=\frac{L}{m\cosh\delta}-m,\qquad
 n_i=\frac{\widehat Q_i}{\widehat R}.
\end{aligned}
\label{eq:o2-reconstruct-1}
\end{equation}

The reconstruction can be checked directly using
\begin{align}
 &J^2+L^2
 =2m^2L\cosh\delta,\notag\\
 &\widehat R^2
 =\frac14m^2\cosh^2\delta\,\sinh^2\delta, \quad \ell^2
 =m^2-b^2.
 \label{eq:o2-audit-identities}
\end{align}
Eliminating $m$ from the first two equations gives
Eq.~\eqref{eq:o2-cubic}. Conversely, the solution of
Eqs.~\eqref{eq:o2-reconstruct-1}
satisfies all three identities. In particular,
\begin{equation}
 \ell=\frac{L^2-J^2}{2Lm\cosh\delta}>0
\end{equation}
selects the outer branch. At zero charge,
Eq.~\eqref{eq:o2-cubic} gives $\cosh\delta=1$,
so the unique nonnegative solution is $\delta=0$.
The remaining thermodynamic quantities are
\begin{align}
 T&=\frac{\ell}{4\pi L}, &\Omega_H&=\frac{b}{2L},\notag\\
 \Phi_i&=\e^{\sqrt3\phi_0/2}n_i\tanh\delta, &P_i&=0,\notag\\
 \Sigma&=-\frac{\sqrt3}{2}m\sinh^2\delta, &A_H&=4S.
 \label{eq:o2-intrinsic-responses}
\end{align}
The scalar term in the extended first law follows from the same mass function. At fixed physical $S,J,Q_i$,

\begin{equation}
 \left.\frac{\partial M}{\partial\phi_0}\right|_{S,J,Q_i}
 =\frac{\sqrt3}{2}\sum_i Q_i\Phi_i
 =-\frac{\Sigma}{2}.
\end{equation}
Thus the charge factors in Eq.~\eqref{eq:o2-dressed} reproduce the extended first law
Eq.~\eqref{eq:general-Smarr}, including variations of $\phi_0$.

The full charged nonextremal domain is
$L>|J|$, $\widehat R>0$, including $J=0$.
As $\widehat R\to0$, $\delta\to0$ and $\Phi_i\to0$,
and the geometry and thermodynamic quantities approach Kerr
independently of the charge direction.
At zero charge, $n_i$ is a redundant parameter.
The finite-area extremal boundary is $L=|J|>0$,
where $T=0$ and $S_{\rm ext}=2\pi|J|$.

\section{Kerr--Newman zero-coupling and balanced embeddings}
\label{sec:embedding}

\subsection{Zero-coupling branch}
\label{sec:zero-coupling}

\subsubsection{Field-equation and complete outer-horizon thermodynamics}
We adopt the following conventions for zero coupling and charge normalization:
\begin{equation}
 \begin{gathered}
 a_1=a_2=0,
 \qquad
 {\cal Q}=\frac1{16\pi}\int_{S^2_\infty}\star F,\\
 {\cal Q}^2=Q_1^2+Q_2^2,
 \qquad
 {\cal Z}_0:=2{\cal Q},
 \qquad
 X_0={\cal Z}_0^2,
 \end{gathered}
 \label{eq:zero-definitions}
\end{equation}
the dilaton decouples and can be chosen constant. The Kerr--Newman solution
and its Christodoulou--Ruffini formula give a normalization check
\cite{NewmanEtAl1965,Carter1968,ChristodoulouRuffini1971}.
Here ${\cal Q}$ is the physical electric flux charge and ${\cal Z}_0$ is the
amplitude appearing in the metric. For ${\cal Q}>0$, set
\begin{equation}
 F_i=\frac{Q_i}{\cal Q}F.
 \label{eq:zero-embedding}
\end{equation}
The coefficients in Eq.~\eqref{eq:zero-embedding} are constant on spacetime.
The seed Maxwell equations and the quadratic Maxwell stress tensor therefore give
\begin{equation}
 \begin{gathered}
 \sum_i{\cal T}[F_i]=\frac{Q_1^2+Q_2^2}{{\cal Q}^2}{\cal T}[F]
 ={\cal T}[F], \quad \dd F_i=0,\\
 \Box\phi=\frac14\sum_i a_i\e^{a_i\phi_0}F_i^2=0,\quad \quad \quad \quad \dd\star F_i=0.
 \end{gathered}
 \label{eq:zero-field-equation-check}
\end{equation}
Thus the Maxwell, scalar, and Einstein equations all hold for this
two-charge embedding of the purely electric Kerr--Newman family. Its neutral limit is
continuous. The internal rotation changes the distribution between $Q_1$
and $Q_2$ while preserving the geometry and mass. The potentials transform
with the charge direction, preserving the complete electromagnetic work.
This branch checks the charge normalization, two-field distribution, and
thermodynamic conventions in the absence of scalar coupling.

With $h=S/\pi$, the outer-horizon quantities are
\begin{align}
 \Omega_H&=\frac{J}{Mh},
 &\Phi_i&=\frac{2Q_i(h+X_0)}{Mh},
 &P_i&=\Psi_i=\Sigma=0,\notag\\
 A_H&=4S,
 &T&=\frac{h^2-X_0^2-4J^2}{8\pi M h^2},
 &\kappa&=2\pi T.\label{eq:zero-ATS}
\end{align}

\subsubsection{First law and Smarr relation}

The first law can be checked term by term using the mass relation
below. Since $\dd h=\dd S/\pi$
and $\dd X_0=8\sum_iQ_i\dd Q_i$, differentiating
Eq.~\eqref{eq:zero-mass} gives

\begin{align}
 &2M\dd M
 ={}\frac{h^2-X_0^2-4J^2}{4h^2}\dd h
 +\frac{2J}{h}\dd J
 +\frac{h+X_0}{2h}\dd X_0,\notag\\
 &\dd M={}T\dd S+\Omega_H\dd J+\sum_i\Phi_i\dd Q_i.
 \label{eq:zero-first-law-check}
\end{align}
Each coefficient agrees with the corresponding horizon quantity in
Eqs.~\eqref{eq:zero-ATS}. Substituting these
quantities also verifies the Smarr relation directly:
\begin{equation}
 2TS+2\Omega_HJ+\sum_i\Phi_iQ_i=\frac{(h+X_0)^2+4J^2}{4Mh}=M.
\label{eq:zero-smarr-check}
\end{equation}
The final equality is Eq.~\eqref{eq:zero-mass}. The mass is independent
of $\phi_0$, consistently with $\Sigma=0$.

\subsubsection{Mass function, regular domain and extremal boundary}

The outer-horizon mass relation is
\begin{equation}
 M^2=\frac{h}{4}+\frac{X_0}{2}
 +\frac{X_0^2+4J^2}{4h}.
 \label{eq:zero-mass}
\end{equation}
The physical domain can be written in entropy or mass variables as
\begin{align}
 &M>0,\qquad h>0,\qquad h\ge\sqrt{X_0^2+4J^2}, \notag\\
 &M>0,\qquad M^2\ge X_0,
 \label{eq:zero-domain-S}\\
 &J^2\le M^2(M^2-X_0),\qquad S=S_+.\notag
\end{align}
where the two algebraic entropy roots are
\begin{equation}
 S_\pm=\pi(2M^2-X_0)
 \mathbin{\pm}2\pi\sqrt{M^2(M^2-X_0)-J^2}.
 \label{eq:zero-entropy}
\end{equation}
The outer horizon has $S=S_+$. The second root gives the positive geometric
inner entropy where the inner horizon is regular, as specified in
Sec.~\ref{app:two-horizon-audit}.

The entropy product is
\begin{equation}
 S_+S_-=\pi^2(X_0^2+4J^2).
 \label{eq:zero-area-product}
\end{equation}

Equality in Eq.~\eqref{eq:zero-domain-S} gives the extremal boundary,
\begin{equation}
 \begin{gathered}
 M_{\rm extremal}^2=\frac{X_0+\sqrt{X_0^2+4J^2}}2,\\
 S_{\rm ext}=\pi\sqrt{X_0^2+4J^2}.
 \end{gathered}
 \label{eq:zero-extremal}
\end{equation}
At $J=0$ and $X_0>0$, this includes static extremal RN with
$M_{\rm extremal}^2=X_0$ and $S_{\rm ext}=\pi X_0$.
Here $\mu_{\rm LR}^2=M^2-X_0$. Substitution of
Eqs.~\eqref{eq:zero-ATS} and \eqref{eq:zero-mass} gives
\begin{equation}
 M^2-X_0=(2TS)^2+\left(\frac{S\Omega_H}{\pi}\right)^2.
 \label{eq:zero-rotating-identity}
\end{equation}
The two charges are independent coordinates, while the geometry depends
on their $O(2)$ norm and the fields remain parallel. This supplies the
zero-coupling check of the two-charge normalization and quadratic relation.

\subsection{Balanced branch}

\subsubsection{Scalar cancellation and field-equation}

We use the action and conventions of
Eqs.~\eqref{eq:general-action}--\eqref{eq:general-Smarr} and take
\begin{equation}
 a_1>0>a_2, \qquad \phi=\phi_0.
 \label{eq:balanced-assumptions}
\end{equation}
Let $(g,F)$ solve Einstein--Maxwell theory with action
\begin{equation}
 I_{\rm EM}=\frac1{16\pi}\int\dd^4x\sqrt{-g}
 \left(R-\frac14F^2\right).
 \label{eq:EM-action}
\end{equation}
The two target fields are defined as
\begin{equation}
F_i=\epsilon_i\sqrt{w_i}\,e^{-a_i\phi_0/2}F,
\qquad \epsilon_i=\pm1,\label{eq:balanced-embedding}
\end{equation}
where
\begin{equation}
a_1>0>a_2,\quad
w_1=-\frac{a_2}{a_1-a_2},\quad
w_2=\frac{a_1}{a_1-a_2} \label{eq:balanced-weights}.
\end{equation}
These two weights satisfy
\begin{equation}
w_1+w_2=1,\qquad a_1w_1+a_2w_2=0 \label{eq:balanced-weights2}.
\end{equation}
They play distinct roles. The first condition guarantees geometric invariance. The total contribution of the two Maxwell fields to the Einstein equations reads
\begin{equation}
\sum_i e^{a_i\phi_0}T_{\mu\nu}[F_i]
=(w_1+w_2)T_{\mu\nu}[F]
=T_{\mu\nu}[F] \label{eq:geometric-invariance}.
\end{equation}
Hence the original metric still satisfies the Einstein equations. 

The second condition ensures that the scalar can remain constant. A constant scalar obeys \(\Box\phi=0\), and the scalar sources generated by the two fields cancel exactly:
\begin{equation}
\sum_i a_i e^{a_i\phi_0}F_i^2
=(a_1w_1+a_2w_2)F^2=0\label{eq:remain-constant}.
\end{equation}
It is the scalar source that cancels here, not the electromagnetic stress tensor; the stress contributions of the two fields still add up. Meanwhile, since each \(F_i\) is a constant multiple of the seed field and \(\phi_0\) is also constant, the Maxwell equations and Bianchi identity for the seed directly yield
\begin{equation}
\mathrm dF_i=0,\qquad
\mathrm d\!\left(e^{a_i\phi_0}\star F_i\right)=0\label{eq:Maxwell-Bianchi}.
\end{equation}
Therefore, the Maxwell equations, Einstein equations and scalar equation are all satisfied separately.

After imposing the field equations, the magnetically charged solutions further require gauge potentials defined consistently on north and south patches. Detailed discussions of the electric potential, magnetic potential, and the full electromagnetic work will be presented in the subsequent thermodynamic subsection.

The fields in Eq.~\eqref{eq:balanced-embedding} correspond to the
$\vartheta_0=0$ opposite-sign case of
Proposition~\ref{prop:constant-scalar-seed-dual}: both target fields
are constant multiples of the seed field, without an additional
electromagnetic duality rotation. For a charged rotating
Kerr--Newman seed with $b\ne0$, the two electromagnetic invariants
satisfy the independence assumption of that classification.
Every opposite-sign map covered by the classification can therefore
be obtained by a common duality rotation and independent choices
of field signs, with the same weights $w_i$.

For parallel fields, the weights can also be determined directly
without assuming rotation. Write $F_i=k_iF$ with real constant
$k_i$, and define $x_i=\e^{a_i\phi_0}k_i^2\ge0$.
The Einstein equation requires the total Maxwell stress tensor
to equal that of the seed, while the scalar equation requires
the two scalar sources to cancel. When $F^2\not\equiv0$,
these conditions give
\begin{equation}
 x_1+x_2=1,\qquad a_1x_1+a_2x_2=0.
 \label{eq:balanced-uniqueness-system}
\end{equation}
The solution is $x_i=w_i$, so
\begin{equation}
 k_i=\epsilon_i\sqrt{w_i}\,\e^{-a_i\phi_0/2},
 \qquad \epsilon_i=\pm1.
\end{equation}
Thus the weights are uniquely fixed, and only the field signs
remain free.

If the seed field is nonzero but $F^2\equiv0$, the scalar equation
is automatically satisfied for a parallel split. The Einstein
equation still requires $x_1+x_2=1$, but the scalar equation
no longer fixes the relative weights. A static dyonic
Reissner--Nordstr\"om field with equal nonzero electric and magnetic
charge magnitudes provides an example: $F^2\equiv0$ while
$F\mathbin{\cdot}\star F\ne0$. The balanced weights still give a
solution, but other nonnegative weights satisfying $x_1+x_2=1$
are also allowed. Additional static combinations involving both
$F$ and $\star F$ are discussed below
Eq.~\eqref{eq:KN-complex-invariant}.

The opposite-sign Einstein--Maxwell truncation was given in
Ref.~\cite{LuYangLu2025MassCharge}.
The field substitution reduces the EMMD equations to the
Einstein--Maxwell equations, so it applies to both static and
rotating seeds. For $a_1=-a_2=\alpha>0$, the weights are
$w_1=w_2=1/2$, and the coefficients multiplying the seed field are
\begin{equation}
 k_1=\frac{\epsilon_1}{\sqrt2}\e^{-\alpha\phi_0/2},
 \qquad
 k_2=\frac{\epsilon_2}{\sqrt2}\e^{\alpha\phi_0/2}.
\end{equation}
Their magnitudes are therefore both $1/\sqrt2$ at $\phi_0=0$;
the exponential factors must be retained at arbitrary $\phi_0$.
Related two-field EMD and STU embeddings, including
Kerr--Newman backgrounds, appear in the one-loop calculations of
Refs.~\cite{KaranPunia2023,KaranPuniaBiswas2025}.

At fixed couplings, scalar boundary value, and field signs, the
seed electric charge $q$ and magnetic charge $p$ can vary
independently within the regular parameter domain. Their relative
distributions between the two Maxwell fields are fixed, so the
four target charges lie on a two-dimensional plane.
Setting $p=0$ leaves one continuous electric-charge parameter
and a fixed ratio $Q_1/Q_2$. Allowing both signs of $q$ gives a
line; fixing its sign selects a ray. The charge ratio follows
from the classical scalar-source cancellation condition.

The common duality angle rotates the pair $(q,p)$.
Once both signed seed charges are included, this angle does not
add a third independent charge parameter. The field signs are
discrete choices and likewise do not increase the charge-space
dimension. For special static seed fields, additional weight and field combinations are available.

\subsubsection{Duality-equivalent representative}
\label{cor:direct-hodge-constant-scalar}

The same-sign constant-scalar map in
Proposition~\ref{prop:constant-scalar-seed-dual} follows by dualizing
the second field of an opposite-sign balanced solution:
\begin{equation}
 \widetilde a_2=-a_2,
 \qquad
 \widetilde F_2:=\e^{a_2\phi}\star F_2.
 \label{eq:direct-hodge-vector-dualization}
\end{equation}
The Lorentzian identities $\star\star F_2=-F_2$ and
${\cal T}[\star F_2]={\cal T}[F_2]$ preserve the Einstein equation,
while the coupling sign change preserves the scalar equation.
With our charge and dual-potential conventions, the transformation acts as
\begin{equation}
 (\widetilde Q_2,\widetilde P_2)=(-P_2,Q_2),
 \qquad
 (\widetilde\Phi_2,\widetilde\Psi_2)=(-\Psi_2,\Phi_2).
 \label{eq:direct-hodge-dual-charge-potential-map}
\end{equation}

The complete electromagnetic work form, including the magnetic patch
terms, is unchanged. The same mass relation, first law, Smarr relation,
and regular domain therefore apply after this charge transformation.
For nonzero real fields the same-sign representative has $r_F=2$;
its charge dimension remains $d_Q=2$. This is a different duality
representation of the balanced thermodynamics.

\subsubsection{Kerr--Newman seed: metric and parameters}
\label{sec:mother}

In the normalization of Eq.~\eqref{eq:EM-action}, the physical electric and magnetic charges of the seed are
\begin{equation}
\mathfrak q=\frac1{16\pi}\int_{S^2_\infty}\star F,\quad \mathfrak p=\frac1{16\pi}\int_{S^2_\infty}F.
\label{eq:balanced-seed-charges1}
\end{equation}
They belong to single-field seeds, rather than the respective \(Q_i, P_i\) of the two Maxwell fields after embedding. The metric depends on the charges through
\begin{equation}
 \mathcal Z^2=4(\mathfrak q^2+\mathfrak p^2).
 \label{eq:balanced-seed-charges2}
\end{equation}
 With fixed \(M,b\), the Kerr–Newman metric remains unchanged so long as \(\mathfrak q^2+\mathfrak p^2\) is invariant, while the electric charge, magnetic charge and their gauge fields may differ.
The dyonic Kerr--Newman functions and regular north/south potentials are
\cite{NewmanEtAl1965,Carter1968}
\begin{equation}
 \rho^2=r^2+b^2\cos^2\theta, \quad \Delta=r^2-2Mr+b^2+\mathcal Z^2.
\label{eq:rho-delta}
\end{equation}
 \begin{align}
 A^{(N/S)}={}&-\frac{4\mathfrak q r}{\rho^2}
 \left(\dd t-b\sin^2\theta\,\dd\varphi\right)\notag\\
 &+\frac{4\mathfrak p\cos\theta}{\rho^2}
 \left[b\,\dd t-(r^2+b^2)\dd\varphi\right]
 \mathbin{\pm}4\mathfrak p\,\dd\varphi,\notag\\
 &F={}\dd A^{(N)}=\dd A^{(S)}.
 \label{eq:KN-potential}
\end{align}
The upper sign is regular on the north axis and the lower sign on the south
axis. On the overlap, $A^{(N)}-A^{(S)}=8\mathfrak p\,\dd\varphi$, so both potentials define the same field strength.
The factor $4$ is fixed by the action and flux-charge normalization:
the asymptotic fluxes are $16\pi\mathfrak q$ and $16\pi\mathfrak p$.

We now check the independence assumption in
Proposition~\ref{prop:constant-scalar-seed-dual}.
We combine the two electromagnetic invariants into
\begin{equation}
 \mathcal I
 :=F^2+\mathrm iF\mathbin{\cdot}\star F
 =-\frac{32(\mathfrak q-\mathrm i\mathfrak p)^2}
 {(r+\mathrm i b\cos\theta)^4}.
 \label{eq:KN-complex-invariant}
\end{equation}
For $b\ne0$ and $\mathfrak q^2+\mathfrak p^2>0$,
$\mathcal I$ is nonzero throughout the regular exterior.
If the two invariants satisfied a nontrivial linear relation
with real constant coefficients on an open region, the complex quantity $\mathcal I$ can only take values along a fixed straight line passing through the origin on the complex plane. Consequently, its phase must be fixed modulo \(\pi\).
The ratio $\mathcal I/\overline{\mathcal I}$ would therefore
be constant. However, direct differentiation gives
\begin{equation*}
 \partial_\theta
 \log\left(\frac{\mathcal I}{\overline{\mathcal I}}\right)
 =\frac{8\mathrm i b r\sin\theta}{\rho^2}.
\end{equation*}
In the black-hole exterior, $r>0$, so this derivative is
nonzero away from the rotation axis. Every open exterior
region contains such points. The two invariants are therefore
linearly independent on every open region of the regular
exterior. This conclusion includes rotating purely electric
and purely magnetic seeds, as well as finite-area extremal
black holes with $b\ne0$.

For a charged static seed, setting $b=0$ gives
\begin{equation}
F^2
 =\frac{32(\mathfrak p^2-\mathfrak q^2)}{r^4},\quad
 F\mathbin{\cdot}\star F
=\frac{64\mathfrak q\mathfrak p}{r^4}.
 \label{eq:KN-static-invariants}
\end{equation}
The scalar equation supplies one real constraint, allowing additional
direct--dual combinations. Equal electric and magnetic magnitudes give
$F^2\equiv0$ but $F\mathbin{\cdot}\star F\ne0$. A null field requires both invariants to vanish simultaneously while the field itself remains nonvanishing.

We next introduce parameters that simplify the horizon
locations and thermodynamic quantities.
For $M>\mathcal Z\ge0$, write
\begin{equation}
 \begin{aligned}
 M&=m\cosh\eta,
 &\mathcal Z&=m\sinh\eta,\\
 m&>0,
 &\eta&\ge0.
 \end{aligned}
 \label{eq:hyperbolic}
\end{equation}
Here $M$ is the physical mass, while
$m=\sqrt{M^2-\mathcal Z^2}$ is an auxiliary parameter.

The seed field is linear in $(\mathfrak q,\mathfrak p)$,
so reversing both charges sends $F$ to $-F$.
If both field signs $\epsilon_i$ are reversed at the same
time, the target fields in
Eq.~\eqref{eq:balanced-embedding} remain unchanged.
We therefore identify
\begin{equation}
 (\mathfrak q,\mathfrak p,\epsilon_1,\epsilon_2)\sim
 (-\mathfrak q,-\mathfrak p,-\epsilon_1,-\epsilon_2).
 \label{eq:balanced-sign-identification}
\end{equation}
This removes duplicate descriptions of the same target
solution. In the nonextremal domain $|b|<m$, define
\begin{equation}
 \ell=\sqrt{m^2-b^2},\quad
 r_\pm=M\pm\ell,\quad
 D_H=r_+^2+b^2.
 \label{eq:horizon-definitions}
\end{equation}
The quantity $D_H$ is the area factor used below.
In the neutral static case, $r_-=0$ is not a regular
inner horizon and is excluded from the two-horizon comparison.
This parametrization not only simplifies the horizon roots but also allows direct verification of the thermodynamic relations.

\subsubsection{Outer-horizon thermodynamics and the charge plane}

Applying Eq.~\eqref{eq:balanced-embedding} to
Eq.~\eqref{eq:KN-potential} gives
\begin{equation}
\begin{aligned}
 &M=m\cosh\eta,\quad \quad \quad J=Mb,\quad \quad \quad \Sigma=0,\\
 &Q_i=\epsilon_i\sqrt{w_i}\,\e^{a_i\phi_0/2}\mathfrak q, \quad
 P_i=\epsilon_i\sqrt{w_i}\,\e^{-a_i\phi_0/2}\mathfrak p.
\end{aligned}
 \label{eq:mother-global-1}
\end{equation}
The horizon quantities are
\begin{equation}
\begin{aligned}
 A_H&=4\pi D_H,
 &S&=\pi D_H,\\
 \Omega_H&=\frac{b}{D_H}, &T&=\frac{\ell}{2\pi D_H}, \quad \quad \kappa=\frac{\ell}{D_H}.
\end{aligned}
 \label{eq:mother-horizon-1}
\end{equation}
These relations satisfy $S=A_H/4,T=\kappa/2\pi$. The gauge-invariant seed and target potentials are
\begin{equation}
\begin{aligned}
 &\Phi=\frac{4\mathfrak q r_+}{D_H},
 &&\Phi_i=\epsilon_i\sqrt{w_i}\,\e^{-a_i\phi_0/2}\Phi,\\
 &\Psi=\frac{4\mathfrak p r_+}{D_H},
 &&\Psi_i=\epsilon_i\sqrt{w_i}\,\e^{a_i\phi_0/2}\Psi.
\label{eq:balanced-dyonic-potential-map}
\end{aligned}
\end{equation}
Define the coefficients
\begin{equation}
 \alpha_i:=\epsilon_i\sqrt{w_i}\,\e^{a_i\phi_0/2},
 \qquad
 \beta_i:=\epsilon_i\sqrt{w_i}\,\e^{-a_i\phi_0/2}.
 \label{eq:balanced-plane-coefficients}
\end{equation}
Then the charge and potential mappings can be compactly written as
\begin{equation}
\begin{aligned}
Q_i=\alpha_i\mathfrak q,\quad
\Phi_i=\beta_i\Phi,\quad
P_i=\beta_i\mathfrak p,\quad
\Psi_i=\alpha_i\Psi,
\end{aligned}
\end{equation}
where $\alpha_i\beta_i=w_i$.
That is, the electric charge \(Q_i\) and its electric potential \(\Phi_i\) carry opposite scalar exponents; the same holds for the magnetic charge \(P_i\) and its magnetic potential \(\Psi_i\).

This distinction comes from the definition of physical charges, not from artificial choices. The field mapping is $F_i=\beta_iF$.
The magnetic charge is directly given by the flux of \(F_i\), hence $P_i=\beta_i\mathfrak p$.
The flux definition for electric charge additionally includes \(e^{a_i\phi_0}\), so
\begin{equation}
Q_i=e^{a_i\phi_0}\beta_i\mathfrak q
=\alpha_i\mathfrak q.
\end{equation}
This follows by substituting the field mapping directly into the flux definition. The electric and magnetic potentials transform according to the definitions of the corresponding gauge potential and dual gauge potential.

These opposite exponents precisely cancel the scalar factors in the work terms. For instance,
\begin{equation}
Q_i\Phi_i=w_i\mathfrak q\Phi,
\qquad
P_i\Psi_i=w_i\mathfrak p\Psi.
\end{equation}
Thus, even though the two fields carry distinct charges separately, their total electromagnetic contributions still match those of the seed. Meanwhile,
\begin{equation}
\begin{aligned}
&A_i^{(N/S)}=\beta_i A^{(N/S)},\\
&A_i^{(N)}-A_i^{(S)}
=8\beta_i\mathfrak p\,\mathrm d\varphi
=8P_i\,\mathrm d\varphi.
\end{aligned}
 \label{eq:balanced-difference}
\end{equation}
This shows that the northern and southern gauge patches of the seed can be transported to the target field via the same embedding map.

With \(a_i,\phi_0,\epsilon_i\) fixed, \(\alpha_i,\beta_i\) are constant coefficients, the four charges satisfy two linear constraints
\begin{equation}
 \alpha_2Q_1-\alpha_1Q_2=0,
 \qquad
 \beta_2P_1-\beta_1P_2=0.
 \label{eq:balanced-dyonic-plane}
\end{equation}
All admissible charge combinations lie on a two dimensional plane within the four dimensional charge coordinate space.
The charge map has rank two, while $F_1\parallel F_2$ gives field rank one.
The radial charge combination invariant under the scalar shift is
\begin{equation}
 \begin{aligned}
 X_{\rm bal}^{\rm dy}
:&=4\sum_i\left(
 \e^{-a_i\phi_0}Q_i^2+\e^{a_i\phi_0}P_i^2\right)\\
 &=4(\mathfrak q^2+\mathfrak p^2)=\mathcal Z^2.
 \end{aligned}
 \label{eq:balanced-dyonic-invariant}
\end{equation}

On the purely electric slice $\mathfrak p=0$, the charge ratios hold
throughout the corresponding rotating electric rays:

\begin{equation}
 \begin{aligned}
 &\frac{Q_1}{Q_2}
 =\frac{\epsilon_1}{\epsilon_2}\sqrt{-\frac{a_2}{a_1}}\,
 \e^{(a_1-a_2)\phi_0/2},\\
 &X_{\rm bal}^{\rm dy}\big|_{\mathfrak p=0}
=4\left(\frac{Q_1^2}{w_1}\right)^{w_1}
 \left(\frac{Q_2^2}{w_2}\right)^{w_2}.
 \end{aligned}
 \label{eq:X-and-ratio}
\end{equation}

At fixed $\phi_0$, this leaves one continuous electric-charge direction,

\begin{equation}
 \dd\ln\left|\frac{Q_1}{Q_2}\right|=0,
 \qquad
 \frac{\dd Q_1}{Q_1}=\frac{\dd Q_2}{Q_2}
 =\frac{\dd\mathfrak q}{\mathfrak q}.
 \label{eq:balanced-fixed-modulus-tangent}
\end{equation}

Changing the scalar boundary value rotates this ray in charge space:

\begin{equation}
 \dd\ln\left|\frac{Q_1}{Q_2}\right|
 =\frac{a_1-a_2}{2}\,\dd\phi_0.
 \label{eq:balanced-modulus-sweep}
\end{equation}

The purely magnetic line follows by exchanging electric and magnetic
quantities and reversing the exponent of $\phi_0$. Each line is a
one-dimensional subset of the full dyonic charge plane.
Thus, at fixed $\phi_0$, the full dyonic branch has two independent
charge amplitudes, while its purely electric slice has one. Changing the scalar boundary value changes the direction of the allowed electric-charge line.

\subsubsection{First law and Smarr relation}
\label{sec:thermodynamics}

Substituting the electromagnetic charge map and potential map, one may directly verify the first law. Equations~\eqref{eq:mother-global-1} and
\eqref{eq:balanced-dyonic-potential-map} give
\begin{equation}
\begin{aligned}
 \Phi_i\dd Q_i
 &=w_i\Phi\left(\dd\mathfrak q
 +\frac12a_i\mathfrak q\,\dd\phi_0\right),\\
 \Psi_i\dd P_i
 &=w_i\Psi\left(\dd\mathfrak p
 -\frac12a_i\mathfrak p\,\dd\phi_0\right),
 \label{eq:balanced-individual-work}
\end{aligned}
\end{equation}
and hence
\begin{equation}
 \sum_i(\Phi_i\dd Q_i+\Psi_i\dd P_i)
 =\Phi\dd\mathfrak q+\Psi\dd\mathfrak p.
 \label{eq:work-collapse}
\end{equation}
The $\dd\phi_0$ terms cancel because $\sum_i a_iw_i=0$, in agreement with
$\Sigma=0$. This includes the magnetic patch contribution.
Taking $(M,b,\mathfrak q,\mathfrak p)$ as independent variables gives
\begin{align}
 &T\dd S
 =\frac{r_+^2\dd M+b(\ell-r_+)\dd b}{D_H}
 -\frac{4r_+}{D_H}
 (\mathfrak q\,\dd\mathfrak q+\mathfrak p\,\dd\mathfrak p),\notag\\
 &\Omega_H\dd J
 =\frac{b^2\dd M+bM\dd b}{D_H}.
 \label{eq:Tds}
\end{align}
Using $r_+=M+\ell$, the sum of all terms in
Eqs.~\eqref{eq:work-collapse}--\eqref{eq:Tds} equals $\dd M$, this proves that the first law holds on the family of Balanced branch.
For the Smarr relation,
\begin{equation}
 \begin{aligned}
 &2TS=\ell,\quad
 2\Omega_HJ=\frac{2Mb^2}{D_H},\\
 &\sum_i(\Phi_iQ_i+\Psi_iP_i)
 =\frac{\mathcal Z^2r_+}{D_H},
 \end{aligned}
\end{equation}
and the horizon equation makes their sum equal to $M$.
Since 
\begin{equation}
Q_i\Phi_i=w_i\mathfrak q\Phi, \quad
P_i\Psi_i=w_i\mathfrak p\Psi, \quad
\sum_i a_iw_i=0.
\end{equation} 
The scalar-shift identity gives
\begin{equation}
\Sigma=\sum_i a_i(Q_i\Phi_i-P_i\Psi_i)=0.
\end{equation}
On the balanced family, 
\begin{equation}
\mu_{\rm LR}^2=M^2-X_{\rm bal}^{\rm dy}=m^2,\:
2TS=\ell, \:
S\Omega_H/\pi=b.
\end{equation} 
Since $\ell^2+b^2=m^2$, the seed relation \eqref{eq:intro-unified-rotating-identity} follows immediately.

\subsubsection{Mass function, regular domain and extremal boundary}

Eliminating $(r_+,b,\mathfrak q,\mathfrak p)$ , define $X=X_{\rm bal}^{\rm dy}$, we hanve
\begin{equation}
 M^2(S,J,Q_i,P_i;\phi_0)
 =\frac{S}{4\pi}+\frac{X}{2}
 +\frac{\pi}{4S}\left(X^2+4J^2\right).
 \label{eq:fundamental}
\end{equation}
This is the Kerr--Newman Christodoulou--Ruffini formula
\cite{ChristodoulouRuffini1971,CveticEtAl2018} with
$\mathcal Z^2=X_{\rm bal}^{\rm dy}$. On a fixed signed balanced plane, the
nonextremal outer branch is
\begin{equation}
 M>0,
 \qquad X>0,
 \qquad S>\pi\sqrt{X^2+4J^2}.
 \label{eq:domain}
\end{equation}
Equality holds in the extremal case. Including that boundary, the equivalent mass
conditions are
\begin{equation}
 \begin{aligned}
 M&>0, &X&>0, \quad \quad  &M^2&\ge X,\\
 J^2&\le M^2(M^2-X), &S&=S_+,
 \end{aligned}
 \label{eq:balanced-mass-domain}
\end{equation}
where
\begin{equation}
 S_\pm=\pi(2M^2-X)
 \mathbin{\pm}2\pi\sqrt{M^2(M^2-X)-J^2}.
 \label{eq:balanced-entropy-branches}
\end{equation}
The outer horizon has entropy $S_+$, and the product is
\begin{equation}
S_+S_-=\pi^2(X^2+4J^2).
\end{equation}
 At extremality,
\begin{equation}
 S_{\rm ext}=\pi\sqrt{X^2+4J^2},\:\:
 M_{\rm ext}^2=\frac{X+\sqrt{X^2+4J^2}}2.
 \label{eq:balanced-extremal}
\end{equation}
As $X_{\rm bal}^{\rm dy}\to0$, every signed balanced plane approaches the
neutral Kerr family. The mass and all potentials and horizon quantities
have the same continuous Kerr limit, while the charge-direction labels
become redundant.

Let $M_{\rm ext}$ be the algebraic extension of
\eqref{eq:fundamental} to the four charge coordinates. It give the
following representative
\begin{equation}
\begin{aligned}
 \frac{\partial M}{\partial S}
 &=\frac1{2M}\left[
 \frac1{4\pi}-\frac{\pi(X^2+4J^2)}{4S^2}\right]=T,\\
 \frac{\partial M}{\partial J}
 &=\frac{\pi J}{MS}=\Omega_H,\\
 \frac{\partial M}{\partial Q_i}
 &=\frac{2\e^{-a_i\phi_0}Q_i}{M}
 \left(1+\frac{\pi X}{S}\right)=\Phi_i,\\
 \frac{\partial M}{\partial P_i}
 &=\frac{2\e^{a_i\phi_0}P_i}{M}
 \left(1+\frac{\pi X}{S}\right)=\Psi_i.
\end{aligned}
 \label{eq:Phi-derivatives}
\end{equation}
The four derivatives in Eq.~\eqref{eq:Phi-derivatives} reproduce the
potentials obtained from the fields. The physical charges obey the two
constraints in Eq.~\eqref{eq:balanced-dyonic-plane}; allowed variations
therefore have two independent directions. Substituting them into the work
term gives Eq.~\eqref{eq:work-collapse}.

\subsubsection{Static and extremal Reissner--Nordstr\"om limits}
\label{subsec:balanced-static-limits}

At $J=0$ and $X>0$, Eq.~\eqref{eq:balanced-extremal} gives the static
extremal RN state, $M_{\rm ext}^2=X$ and $S_{\rm ext}=\pi X$.
Eq.~\eqref{eq:hyperbolic} reaches the static
charged extremal Reissner--Nordstr\"om solution through the scaling limit
\begin{equation}
\begin{aligned}
 &b=0,\qquad \quad \:\: m\to0,\qquad \quad \: \:\eta\to\infty,\\
 &m\cosh\eta\to\mathcal Z,\qquad m\sinh\eta\to\mathcal Z>0.
 \end{aligned}
\end{equation}
The limit has $M=\mathcal Z$, $r_+=\mathcal Z$,
$A_H=4\pi\mathcal Z^2$, and $T=\Omega_H=0$. It is included in the extremal
boundary of the intrinsic mass relation \eqref{eq:fundamental}. Appendix~\ref{app:regularity}
examines both this static limit and the rotating extremal horizon.

\section{Rasheed--Larsen direct--dual thermodynamics}
\label{sec:hodge}

\subsection{Embedding and on-shell duality}

Here we take an existing single-field Rasheed–Larsen black hole and transform it into an EMMD black hole with two Maxwell fields, while preserving the original metric and the full scalar profile unchanged. The first target field is taken as a constant multiple of the seed field, and the second as a constant multiple of its scalar-weighted dual. Calculations verify that this assignment leaves the total source terms in the Einstein equations and scalar equation unaltered.

It has an important distinction from the constant-scalar balanced  branch: in the equilibrium branch, scalar sources cancel one another; here, by contrast, the scalar source of the seed is retained, so that the originally position-dependent scalar field continues to satisfy the equations.

We now introduce magnetic monopole charge while keeping the physical NUT charge vanishing. The seed is Rasheed's rotating dyonic Kaluza--Klein solution, using Larsen's thermodynamic parameters and extremal limits
\cite{Rasheed1995,Larsen2000}. Let $(g,\phi,H)$ solve the one-field theory
\begin{equation}
 \begin{aligned}
 I_{\rm seed}={}&\frac1{16\pi}\int\sqrt{-g}\left[
 R-\frac12(\partial\phi)^2-\frac14\e^{a\phi}H^2\right],\\
 a={}&\sqrt3.
 \end{aligned}
 \label{eq:hodge-seed-action}
\end{equation}
For $c=\cos\vartheta$, $s=\sin\vartheta$, $\vartheta$ is the field distribution angle. Define the two oppositely coupled
fields by
\begin{equation}
 F_1=cH,\:\:  F_2=s\e^{a\phi}\star H,\:\:  
 a_1=a,  \:\:   a_2=-a.
 \label{eq:hodge-fields}
\end{equation}
Varying $\vartheta$ redistributes the seed field $H$ and its weighted dual
$\e^{a\phi}\star H$ between the two Maxwell fields. At $s=0$ only $F_1$ is present; at $c=0$
only the dual field $F_2$ is present. Both endpoints are generally dyonic.
For a purely electric seed, the first field carries electric monopole charge, and the second carries magnetic monopole charge.

The seed equations $\dd H=0=\dd(\e^{a\phi}\star H)$ and $\star^2=-1$
imply both target Bianchi identities and Maxwell equations. The map is an
on-shell electromagnetic duality rotation in $SO(2)\subset Sp(4,\mathbb R)$. Here “on-shell” means that the known seed satisfies the field equations, and the target fields also satisfy the field equations after the prescribed electromagnetic transformation. This is not an ordinary gauge transformation. A gauge transformation alters the expression of the potential without changing the field strength; here, what are modified are the two target field strengths and their electromagnetic charges.
To display it, set $G_i=\e^{a_i\phi}\star F_i$, form
$\mathcal X=(F_1,F_2,G_1,G_2)^{\mathsf T}$, and define
\begin{equation}
\Omega=\begin{pmatrix}0&I_2\\-I_2&0\end{pmatrix},\:
 K=\begin{pmatrix}0&1\\1&0\end{pmatrix}.
 \label{eq:hodge-symplectic-rotation}
\end{equation}
Matrix $K$ swaps the labels of the two fields, allowing the first field to pair with the dual of the second. Here \(\Omega\) is the matrix describing this pairing structure, and
\begin{equation}
\mathcal R_\vartheta=\begin{pmatrix}cI_2&sK\\-sK&cI_2\end{pmatrix},\qquad
 \mathcal R_\vartheta^{\mathsf T}\Omega\mathcal R_\vartheta=\Omega,
\end{equation}
indicates that this transformation preserves the symplectic structure among the electromagnetic variables. 
In the Gaillard--Zumino description \cite{GaillardZumino1981},
the transformation $\mathcal R_\vartheta$ acting on the seed reads
\begin{equation}
\begin{aligned}
 &(H,0,G,0)^{\mathsf T}\longmapsto(cH,sG,cG,-sH)^{\mathsf T},\\
 & G=\e^{a\phi}\star H.
 \end{aligned}
 \label{eq:hodge-symplectic-orbit}
\end{equation}
The relation $G_i=\e^{a_i\phi}\star F_i$ is preserved with the same dilaton.
Dualizing just the second target field gives
\begin{equation}
 \widetilde F_2:=\e^{-a\phi}\star F_2=-sH,
 \quad (F_1,\widetilde F_2)=(cH,-sH).
 \label{eq:hodge-dualized-split}
\end{equation}
In this dualized representation, both fields become constant multiples of the same $H$, which is an ordinary internal split. The equivalence holds for the classical
field equations. Magnetic solutions are described by global field strengths
and gauge potentials on overlapping patches.

Direct source matching gives
\begin{align}
 &\e^{a\phi}{\cal T}[F_1]+\e^{-a\phi}{\cal T}[F_2]
 =\e^{a\phi}{\cal T}[H],\notag\\
 &a\e^{a\phi}F_1^2-a\e^{-a\phi}F_2^2
 =a\e^{a\phi}H^2,
 \label{eq:hodge-source-check}
\end{align}
The sum of the squares of the two coefficients $c^2+s^2$ equals unity, ensuring that both the total stress tensor and total scalar source remain identical to those of the seed. Consequently, the target fields still satisfy all EMMD equations. Instead of cancelling the scalar sources, we retain the original scalar source of the seed.
For such seed solutions, there are relatively few constraints arising purely from the Maxwell equations. Starting from the known
Rasheed--Larsen geometry, the following calculation gives the target charge
constraint and two-field work term.

\subsection{Charge surface and electromagnetic work}
\label{subsec:hodge-charge-work}

The field map determines how the seed electric and magnetic charges
are distributed between the two target fields. We define the signed
seed charges by
\begin{equation}
 \mathsf q=\frac1{16\pi}\int_{S^2_\infty}\e^{a\phi}\star H,
 \qquad
 \mathsf p=\frac1{16\pi}\int_{S^2_\infty}H.
 \label{eq:hodge-seed-charges-def}
\end{equation}
Using Eq.~\eqref{eq:hodge-fields}, the target charges are
\begin{equation}
 Q_1=c\mathsf q,\quad P_1=c\mathsf p,\quad
 Q_2=-s\mathsf p,\quad P_2=s\mathsf q.
\end{equation}
The seed electric charge, magnetic charge, and distribution angle determine the four target electromagnetic charges. These charges always satisfy \(Q_1Q_2+P_1P_2=0\), so they cannot vary independently. Away from the zero-charge point, this constraint defines a smooth
three-dimensional charge surface, as shown in Appendix~\ref{app:charge-space}. At fixed scalar boundary value, this gives three independent charge parameters, locally described by $(\mathsf q,\mathsf p,\vartheta)$.

Each set of nonzero target electric and magnetic charges corresponds to exactly two sets of seed charges and distribution angle parameters; these two sets of parameters describe the same target charge configuration.
They are related by
\begin{equation}
 (\mathsf q,\mathsf p,\vartheta)
 \longmapsto
 (-\mathsf q,-\mathsf p,\vartheta+\pi).
 \label{eq:hodge-double-cover}
\end{equation}
Indeed, increasing $\vartheta$ by $\pi$ reverses both $c$ and $s$,
so reversing the seed charges at the same time leaves all four
target charges unchanged. The zero-charge point is excluded
because every angle gives the same zero-charge configuration there.

We next determine the electromagnetic work.
Use a consistent asymptotic reference for the seed potentials
$\Phi$ and $\Psi$, including the required magnetic gauge-patch
contributions. At fixed remaining seed parameters, choose the
reference so that simultaneous reversal of both charges gives
\begin{equation}
 (\Phi,\Psi)(-\mathsf q,-\mathsf p)
 =-(\Phi,\Psi)(\mathsf q,\mathsf p).
\end{equation}
The target potentials are
\begin{equation}
 (\Phi_1,\Psi_1,\Phi_2,\Psi_2)
 =(c\Phi,c\Psi,-s\Psi,s\Phi).
 \label{eq:hodge-potential-map}
\end{equation}
Under Eq.~\eqref{eq:hodge-double-cover}, both the seed potentials
and the coefficients $c,s$ reverse sign. The target potentials
therefore remain unchanged and do not depend on which of the two
equivalent seed descriptions is used.

Substituting the charge and potential maps gives
\begin{equation}
 \sum_i\left(\Phi_i\dd Q_i+\Psi_i\dd P_i\right)
 =\Phi\dd\mathsf q+\Psi\dd\mathsf p.
 \label{eq:hodge-work-map}
\end{equation}
The total electromagnetic work of the two target fields is
therefore exactly the seed electromagnetic work.
In particular, changing only $\vartheta$ at fixed seed data
changes the target charge distribution but gives zero total
electromagnetic work.
Direct substitution gives
\begin{equation}
 \begin{aligned}
 \sum_i(\Phi_i\dd Q_i+\Psi_i\dd P_i)
 &=(c^2+s^2)(\Phi\dd\mathsf q+\Psi\dd\mathsf p)\\
 &\quad+(c\dd c+s\dd s)
 (\mathsf q\Phi+\mathsf p\Psi).
 \label{eq:hodge-work-expansion}
 \end{aligned}
\end{equation}
Since $c^2+s^2=1$ and $c\dd c+s\dd s=0$, this reduces to
Eq.~\eqref{eq:hodge-work-map}.

The first law must be evaluated for variations that preserve
the charge constraint:
\begin{equation*}
 Q_2\dd Q_1+P_2\dd P_1
 +Q_1\dd Q_2+P_1\dd P_2=0.
\end{equation*}
It therefore fixes the total work for these allowed variations,
rather than four independent charge derivatives.

\subsection{Thermodynamics and scalar boundary value}

The Kaluza--Klein family is a one-dilaton, axion-free
specialization of rotating STU black holes
\cite{ChowCompere2014}. We use the Rasheed--Larsen formulas on
the nonextremal domain
\begin{equation}
 m>0,\quad |b|<m,\quad
 \mathscr P,\mathscr Q\ge 2m,\quad
 d=\sqrt{m^2-b^2}>0 .
 \label{eq:rasheed-domain}
\end{equation}
Here $\mathscr P$ and $\mathscr Q$ are positive auxiliary
parameters of the Rasheed--Larsen solution, rather than the
physical electromagnetic charges. The signs of the physical
electric and magnetic charges $\mathsf q$ and $\mathsf p$ are
chosen independently. At $\phi_0=0$, their magnitudes are
related to the auxiliary parameters by
\begin{equation}
 \mathsf q^2=
 \frac{\mathscr Q(\mathscr Q^2-4m^2)}
 {16(\mathscr P+\mathscr Q)},
 \qquad
 \mathsf p^2=
 \frac{\mathscr P(\mathscr P^2-4m^2)}
 {16(\mathscr P+\mathscr Q)} .
 \label{eq:rasheed-seed-charges}
\end{equation}

The original Rasheed--Larsen formulas use a different
normalization of the electromagnetic charges and their
conjugate potentials
\cite{Rasheed1995,Larsen2000}. We denote those quantities by
$(Q_{\rm R},P_{\rm R},\Phi_{\rm R},\Psi_{\rm R})$.
The flux-charge normalization adopted in this paper is related
to the Rasheed--Larsen convention by
\begin{equation}
 Q_{\rm R}=2\mathsf q,\quad
 P_{\rm R}=2\mathsf p,\quad
 \Phi=2\Phi_{\rm R},\quad
 \Psi=2\Psi_{\rm R}.
 \label{eq:rasheed-normalization}
\end{equation}
With this rescaling, the electromagnetic work term keeps the
standard first-law form:
\begin{equation}
 \Phi\,\dd\mathsf q+\Psi\,\dd\mathsf p
 =
 \Phi_{\rm R}\,\dd Q_{\rm R}
 +\Psi_{\rm R}\,\dd P_{\rm R}.
 \label{eq:rasheed-work-normalization}
\end{equation}
At $\phi_0=0$, the mass, angular momentum, and outer-horizon
quantities are
\begin{equation}
 \begin{aligned}
 &M=\frac{\mathscr P+\mathscr Q}{4},
 \qquad
 J=\frac{\sqrt{\mathscr P\mathscr Q}
 (\mathscr P\mathscr Q+4m^2)b}
 {4m(\mathscr P+\mathscr Q)},\\
 &S=\pi\sqrt{\mathscr P\mathscr Q}
 \left[
 m+\frac{\mathscr P\mathscr Q+4m^2}
 {2m(\mathscr P+\mathscr Q)}d
 \right],\\
 &\frac1T=
 \frac{\pi\sqrt{\mathscr P\mathscr Q}}{m}
 \left[
 \frac{\mathscr P\mathscr Q+4m^2}{\mathscr P+\mathscr Q}
 +\frac{2m^2}{d}
 \right],\\
 &\Omega_H=
 \frac{2m(\mathscr P+\mathscr Q)b}
 {\sqrt{\mathscr P\mathscr Q}
 \left[
 2m^2(\mathscr P+\mathscr Q)
 +(\mathscr P\mathscr Q+4m^2)d
 \right]},\\
 &\Sigma
 =\frac{\sqrt3}{4}(\mathscr Q-\mathscr P),\qquad A_H=4S,\qquad \kappa=2\pi T.
 \end{aligned}
 \label{eq:rasheed-Omega-Sigma}
\end{equation}

When $\mathscr P>2m$ and $\mathscr Q>2m$, both seed charges
are nonzero. The electric and magnetic potentials are
\begin{equation}
 \frac{\Phi}{T}
 =\sgn(\mathsf q)\frac{\pi}{m}
 \sqrt{\frac{\mathscr P(\mathscr Q^2-4m^2)}
 {\mathscr P+\mathscr Q}}
 \left(\mathscr P+\frac{2m^2}{d}\right).
 \label{eq:rasheed-Phi}
\end{equation}
\begin{equation}
 \frac{\Psi}{T}
 =\sgn(\mathsf p)\frac{\pi}{m}
 \sqrt{\frac{\mathscr Q(\mathscr P^2-4m^2)}
 {\mathscr P+\mathscr Q}}
 \left(\mathscr Q+\frac{2m^2}{d}\right).
 \label{eq:rasheed-Psi}
\end{equation}
The charge axes are reached continuously.
At $\mathscr Q=2m$, the electric charge and its nonextremal
potential vanish. At $\mathscr P=2m$, the magnetic charge
and its nonextremal potential vanish.
Equations~\eqref{eq:hodge-charge-map} and
\eqref{eq:hodge-potential-map} then give the four target
charges and their potentials.

We check the seed first law by differentiating with respect
to the four independent parameters
$x\in\{m,b,\mathscr P,\mathscr Q\}$.
On each fixed-sign component with both seed charges nonzero,
direct differentiation gives
\begin{equation}
 \begin{aligned}
 &T\partial_x S+\Omega_H\partial_xJ
 +\Phi\partial_x\mathsf q+\Psi\partial_x\mathsf p\\
 &\qquad=
 \begin{cases}
 0,&x=m,b,\\
 1/4,&x=\mathscr P,\mathscr Q.
 \end{cases}
 \end{aligned}
 \label{eq:rasheed-first-law-check}
\end{equation}
These are precisely the coefficients of the mass differential,
\begin{equation}
 \dd M=\frac{\dd\mathscr P+\dd\mathscr Q}{4}.
 \label{eq:rasheed-mass-differential}
\end{equation}
The charge signs cancel in $\Phi\dd\mathsf q$ and
$\Psi\dd\mathsf p$, so the calculation applies to every
nonzero-charge sign component. The nonextremal charge-axis
values are obtained by continuity in the signed physical
charges.

The algebraic products give
\begin{equation}
 2TS=d,\quad
 \Phi\mathsf q+\Psi\mathsf p+2\Omega_HJ
 =\frac{\mathscr P+\mathscr Q}{4}-d,
 \label{eq:rasheed-Smarr-check}
\end{equation}
and
\begin{equation}
 \Phi\mathsf q-\Psi\mathsf p
 =\frac{\mathscr Q-\mathscr P}{4}
 =\frac{\Sigma}{\sqrt3}.
 \label{eq:rasheed-scalar-shift-check}
\end{equation}
Thus both the seed Smarr relation and the scalar-shift
identity hold explicitly.

Equation~\eqref{eq:hodge-work-map} carries the first-law
calculation to the target charge surface.
Substituting the charge and potential maps into the scalar
combination gives
\begin{equation}
 \sum_i a_i(Q_i\Phi_i-P_i\Psi_i)
 =a(\Phi\mathsf q-\Psi\mathsf p)=\Sigma.
 \label{eq:rasheed-target-scalar-check}
\end{equation}
Including the scalar boundary variation derived below,
the target first law, Smarr relation, and scalar-shift
identity are
\begin{equation}
 \begin{aligned}
 \dd M={}&T\dd S+\Omega_H\dd J\\
 &+\sum_i(\Phi_i\dd Q_i+\Psi_i\dd P_i)
 -\frac12\Sigma\dd\phi_0.
 \end{aligned}
 \label{eq:hodge-first-law}
\end{equation}
\begin{equation}
 M=2TS+2\Omega_HJ
 +\sum_i(\Phi_iQ_i+\Psi_iP_i).
 \label{eq:hodge-Smarr}
\end{equation}
\begin{equation}
 \Sigma=\sum_i a_i(Q_i\Phi_i-P_i\Psi_i).
 \label{eq:hodge-Ward}
\end{equation}
The variations in Eq.~\eqref{eq:hodge-first-law} must
preserve the charge constraint.

The scalar shift restores an arbitrary boundary value
$\phi_0$. To use the zero-boundary-value seed formulas,
define
\begin{equation}
 \widehat{\mathsf q}
 =\e^{-a\phi_0/2}\mathsf q,
 \qquad
 \widehat{\mathsf p}
 =\e^{a\phi_0/2}\mathsf p.
 \label{eq:rasheed-scaled-seed-charges}
\end{equation}
The corresponding potentials satisfy
\begin{equation}
 \widehat\Phi=\e^{a\phi_0/2}\Phi,
 \qquad
 \widehat\Psi=\e^{-a\phi_0/2}\Psi.
 \label{eq:rasheed-scaled-seed-potentials}
\end{equation}
The hatted quantities are those used in the
$\phi_0=0$ formulas; the unhatted quantities are the
physical seed charges and potentials at the chosen
boundary value. Their products are unchanged:
\begin{equation}
 \widehat\Phi\,\widehat{\mathsf q}=\Phi\mathsf q,
 \qquad
 \widehat\Psi\,\widehat{\mathsf p}=\Psi\mathsf p.
 \label{eq:rasheed-scalar-invariant-products}
\end{equation}

The target charge constraint is independent of $\phi_0$.
On the nonzero charge surface, holding all physical target
charges fixed holds $\mathsf q$, $\mathsf p$, and
$\vartheta$ fixed in either local seed description.
The scalar factors therefore give
\begin{equation}
 \left.\frac{\partial\widehat{\mathsf q}}{\partial\phi_0}
 \right|_{Q_i,P_i}
 =-\frac a2\widehat{\mathsf q},
 \qquad
 \left.\frac{\partial\widehat{\mathsf p}}{\partial\phi_0}
 \right|_{Q_i,P_i}
 =\frac a2\widehat{\mathsf p}.
 \label{eq:rasheed-fixed-charge-scalar-variation}
\end{equation}
Using the seed first law at fixed $S$ and $J$, we obtain
\begin{equation}
 \left.\frac{\partial M}{\partial\phi_0}
 \right|_{S,J,Q_i,P_i}
 =-\frac a2
 \left(
 \widehat\Phi\,\widehat{\mathsf q}
 -\widehat\Psi\,\widehat{\mathsf p}
 \right)
 =-\frac{\Sigma}{2}.
 \label{eq:rasheed-scalar-boundary-response}
\end{equation}
This verifies the scalar term in
Eq.~\eqref{eq:hodge-first-law} on the nonextremal branch.

\subsection{Unique mass function on the physical outer branch}
\label{subsec:rl-unique-mass}

The preceding subsection gives the thermodynamic quantities in terms
of the Rasheed--Larsen seed parameters. We now replace these parameters
with the outer-horizon entropy, angular momentum, and allowed charges.
The result is a single-valued mass function on the regular
nonextremal outer branch. At fixed scalar boundary value, the target
charges lie on the nonzero surface $\mathcal N^\times$, where
$Q_1Q_2+P_1P_2=0$ (See Appendix~\ref{app:rl-mass-proof} and
Appendix~\ref{app:rl-reconstruction}).

We first set $\phi_0=0$. The seed geometry and thermodynamic formulas
are those of Refs.~\cite{Rasheed1995,Larsen2000}.
S\'arosi expressed the Kaluza--Klein scalar cubic and the
$F$-invariant in STU conventions \cite{Sarosi2016FInvariant}.
The conversion to our charges is
\begin{equation}
 Q^1_{\rm STU}=4\mathsf q,\quad
 P^1_{\rm STU}=4\mathsf p,\quad
 \Sigma_{1,\rm STU}=-\frac{2\Sigma}{\sqrt3}.
 \label{eq:rasheed-stu-dictionary}
\end{equation}
The mass and seed charges determine the scalar charge through
\begin{equation}
 \frac{4\mathsf q^2}{\Sigma+\sqrt3M}
 +\frac{4\mathsf p^2}{\Sigma-\sqrt3M}
 =\frac{2\Sigma}{3}.
 \label{eq:rasheed-rational}
\end{equation}
The physical root is selected in the interval
\begin{equation}
 -\sqrt3M<\Sigma<\sqrt3M.
 \label{eq:rasheed-scalar-interval}
\end{equation}
Within this interval, both denominators are nonzero, and the
rational equation is equivalent to
\begin{equation}
 \begin{aligned}
 \Sigma^3-3(M^2+2\mathsf p^2+2\mathsf q^2)\Sigma
 +6\sqrt3M(\mathsf q^2-\mathsf p^2)=0.
 \end{aligned}
 \label{eq:rasheed-cubic}
\end{equation}
With Eq.~\eqref{eq:rasheed-stu-dictionary}, this agrees with the
scalar cubic of Ref.~\cite{Sarosi2016FInvariant}.
On a one-charge axis, clearing the denominators can add roots
at $\Sigma=\pm\sqrt3M$. We therefore select the physical root
using Eq.~\eqref{eq:rasheed-rational} and the open interval
\eqref{eq:rasheed-scalar-interval}, rather than the cubic alone.

Define
\begin{equation}
 \begin{aligned}
 \sigma&=\frac{\Sigma}{\sqrt3},\qquad U=M^2-\sigma^2,\\
 V&=M^2+3\sigma^2-4(\mathsf q^2+\mathsf p^2)=m^2,
 \end{aligned}
 \label{eq:rasheed-UV}
\end{equation}
and
\begin{equation}
 \begin{aligned}
 &\mathcal F=UV,\qquad L=\frac{S}{2\pi},\qquad
 A=16\mathsf p^2\mathsf q^2,\\
 &M_0=(|\mathsf q|^{2/3}+|\mathsf p|^{2/3})^{3/2}.
 \end{aligned}
 \label{eq:rasheed-intrinsic-definitions}
\end{equation}
Here $A$ has dimensions of length to the fourth power and is not
the horizon area $A_H$, which has dimensions of length squared.
The threshold $M_0$ is defined at fixed seed charges; the allowed
rotation is restricted further by the outer-horizon condition.

On the scalar-constraint surface, the STU invariants become
\cite{Sarosi2016FInvariant,ChowCompere2014}
\begin{equation}
 F_{\rm STU}=\mathcal F+A,\qquad \Delta_{\rm STU}=-A.
 \label{eq:rasheed-STU-invariants}
\end{equation}
The outer-horizon entropy $S$ can therefore be expressed through the intrinsic variable \(L\)
\cite{ChowCompere2014},
\begin{equation}
 L=\sqrt{\mathcal F}+\sqrt{\mathcal F+A-J^2},\quad
S=2\pi L.
 \label{eq:rasheed-intrinsic-entropy}
\end{equation}
Both square roots are positive on the nonextremal outer branch.
This selects the domain
\begin{equation}
 L>0,\qquad L^2>|J^2-A|.
 \label{eq:rasheed-intrinsic-domain}
\end{equation}
For every state in this domain, the mass is the unique solution
with $M>M_0$ of
\begin{equation}
 \mathcal F(M)=\frac{(L^2+J^2-A)^2}{4L^2}.
 \label{eq:rasheed-intrinsic-mass}
\end{equation}
To evaluate the left-hand side, choose a trial mass $M>M_0$,
solve Eq.~\eqref{eq:rasheed-rational} for its unique physical
scalar root, and use that root in $\mathcal F=UV$.
Along this root, $\mathcal F(M)$ increases strictly from zero
to infinity as $M$ runs from $M_0$ to infinity. The mass equation
therefore has exactly one solution. 

The construction remains valid on the nonextremal one-charge
boundaries of the charge space, where one charge vanishes.
In the neutral limit, the scalar and electromagnetic charges
disappear and the solution continuously reduces to the Kerr
family.

Once $M$ and $\Sigma$ are known, the seed parameters are recovered as
\begin{equation}
 \mathscr Q=2(M+\sigma),\quad \mathscr P=2(M-\sigma),
 \label{eq:rasheed-reconstruct-1}
\end{equation}
\begin{equation}
 m^2=V=M^2+\Sigma^2-4(\mathsf p^2+\mathsf q^2),
 \label{eq:rasheed-reconstruct-m}
\end{equation}
\begin{equation}
 b^2=\frac{J^2m^2}{\mathcal F+A},\qquad
 \operatorname{sgn}b=\operatorname{sgn}J.
 \label{eq:rasheed-reconstruct-b}
\end{equation}
We take $m>0$ and set $b=0$ when $J=0$.
The identity connecting the denominator in
Eq.~\eqref{eq:rasheed-reconstruct-b} to the charge factors is
proved in Appendix~\ref{app:rl-reconstruction}.
The reconstructed parameters satisfy
\begin{equation}
 M\pm\sigma>0,\quad m>0,\quad |b|<m,\quad
 \mathscr P,\mathscr Q\ge2m.
 \label{eq:rasheed-physical-root}
\end{equation}
Thus the mass equation selects a regular member of the seed family,
not merely a root of an algebraic relation. The scalar charge is
fixed by the other state variables and is not an additional
independent thermodynamic charge.

We next express the result in the target charges.
At $\phi_0=0$, the charge map gives
\begin{equation}
 \mathsf q^2=Q_1^2+P_2^2,\qquad
 \mathsf p^2=P_1^2+Q_2^2.
 \label{eq:hodge-intrinsic-seed-squares}
\end{equation}
These squares are unchanged by the two equivalent parameter
choices in Eq.~\eqref{eq:hodge-double-cover}. Hence the mass is
well defined on the nonzero charge surface $\mathcal N^\times$.
For an arbitrary scalar boundary value, replace the seed-charge
squares in the mass and scalar equations by
\begin{equation}
 \widehat{\mathsf q}^{\,2}=\e^{-a\phi_0}(Q_1^2+P_2^2),\quad
 \widehat{\mathsf p}^{\,2}=\e^{a\phi_0}(P_1^2+Q_2^2).
 \label{eq:hodge-dressed-seed-squares}
\end{equation}
The corresponding charge combination and mass threshold are
\begin{equation}
 \begin{aligned}
 A&=16\widehat{\mathsf p}^{\,2}\widehat{\mathsf q}^{\,2},\\
 M_0&=\left[
 (\widehat{\mathsf q}^{\,2})^{1/3}
 +(\widehat{\mathsf p}^{\,2})^{1/3}
 \right]^{3/2}.
 \end{aligned}
 \label{eq:hodge-dressed-mass-threshold}
\end{equation}
The scalar factors cancel in $A$, so $A$ is independent of
$\phi_0$ at fixed target charges. The mass threshold generally
is not. The target charge norm satisfies
\begin{equation}
 \sum_i\left(\e^{-a_i\phi_0}Q_i^2+\e^{a_i\phi_0}P_i^2\right)
 =\widehat{\mathsf q}^{\,2}+\widehat{\mathsf p}^{\,2}.
 \label{eq:hodge-dressed-norm}
\end{equation}
Together with $S$ and $J$, these weighted charges determine the
single-valued mass $M(S,J,Q_i,P_i;\phi_0)$ on the allowed
three-dimensional charge surface.

The reconstructed mass and scalar charge also give
\begin{equation}
 \mu_{\rm LR}^2=M^2+\Sigma^2
 -4(\widehat{\mathsf q}^{\,2}+\widehat{\mathsf p}^{\,2})=m^2.
 \label{eq:hodge-mu}
\end{equation}
Independently, the seed horizon quantities satisfy
\begin{equation}
 2TS=d=\sqrt{m^2-b^2},\qquad \frac{S\Omega_H}{\pi}=b.
 \label{eq:rasheed-outer-products}
\end{equation}
Their sum of squares is therefore $d^2+b^2=m^2$, which verifies
Eq.~\eqref{eq:intro-unified-rotating-identity} on the outer branch
and by continuity at its regular finite-area extremal limits.
The fast and slow extremal boundaries are distinct; reaching
$M=M_0$ is not the condition for every rotating extremal state.
Their relation to the mass inversion is explained in
Appendix~\ref{app:rl-reconstruction}.

\subsection{Extremal limits and the purely electric subset}
\label{subsec:rl-extremal-boundaries}

The seed-charge formulas in this subsection use $\phi_0=0$. At a general boundary value, use the weighted seed charges in Eq.~\eqref{eq:hodge-dressed-seed-squares} and transform their potentials with the inverse scalar factors.

The boundary of Eq.~\eqref{eq:rasheed-intrinsic-domain} has two finite-area
parts. The fast extremal limit has $L^2=J^2-A>0$. The slow extremal limit
has $L^2=A-J^2>0$ and requires ${\mathsf p}{\mathsf q}\ne0$.
They meet at $L=0$, $J^2=A$, where the outer area vanishes.
Within the nonextremal outer branch, $J^2=A$ instead marks zero inner area;
Sec.~\ref{sec:two-horizon} accounts for this distinction.
On the fast boundary, $b=\epsilon m$ with $\epsilon=\pm1$,
\begin{equation}
\begin{aligned}
 &T=0, \quad \Omega_H=\frac{\epsilon}{\sqrt{\mathscr P\mathscr Q}},  &&S=2\pi\sqrt{J^2-16{\mathsf p}^2{\mathsf q}^2},\\
&\Phi=\frac{4{\mathsf q}}{\mathscr Q}, \quad  \Psi=\frac{4{\mathsf p}}{\mathscr P}, &&|J|>4|{\mathsf p}{\mathsf q}|.
\label{eq:rasheed-fast}
\end{aligned}
\end{equation}
The slow extremal limit is obtained by the scaling
\begin{equation}
 m\to0,\quad b/m\to\eta,\quad |\eta|<1,\quad
 \mathscr P,\mathscr Q>0\ \hbox{fixed}.
 \label{eq:rasheed-slow-limit}
\end{equation}
Writing $Y=|{\mathsf q}|^{2/3}+|{\mathsf p}|^{2/3}$, the limiting quantities are
\begin{equation}
\begin{aligned}
 &\Sigma=\sqrt3\sqrt Y
 (|{\mathsf q}|^{2/3}-|{\mathsf p}|^{2/3}),  \\ 
 &S=2\pi\sqrt{16{\mathsf p}^2{\mathsf q}^2-J^2}, &\Omega_H=0&, \\
 &M=Y^{3/2},  \quad \quad  J=4\eta|{\mathsf p}{\mathsf q}|, &T=0&. 
 \label{eq:rasheed-slow-1}\\
 \end{aligned}
\end{equation}
\begin{equation}
 \Phi=\frac{\sgn({\mathsf q})\sqrt Y}{|{\mathsf q}|^{1/3}},\quad
 \Psi=\frac{\sgn({\mathsf p})\sqrt Y}{|{\mathsf p}|^{1/3}},\quad
 |J|<4|{\mathsf p}{\mathsf q}|.
 \label{eq:rasheed-slow-2}
\end{equation}
The endpoint $|J|=4|{\mathsf p}{\mathsf q}|$ has zero area and is outside
both finite-area branches. 
On the fast branch, the angular momentum exceeds the charge-dependent threshold, the horizon angular velocity $\Omega_H$ is nonzero, $\mu_{\rm LR}>0$ and the mass remains above the charge-only threshold \(M_0\). On the slow branch, the angular momentum lies below this threshold and the mass equals \(M_0\). These black holes can carry nonzero total angular momentum even though their horizon angular velocity vanishes, that is $J\ne0$ with $0<|\eta|<1$, while
$\Omega_H=0$ and $\mu_{\rm LR}=0$. 

For $cs\ne0$ and a nonzero seed, $H$ and $\star H$ are independent over
the reals. The field rank is two in the original electric frame and becomes
one after Eq.~\eqref{eq:hodge-dualized-split}. The purely electric subset
of this particular map satisfies
\begin{equation}
 P_1=P_2=0\quad\Longrightarrow\quad Q_1Q_2=0.
 \label{eq:hodge-pure-electric-no-go}
\end{equation}
This Rasheed–Larsen direct-duality embedding cannot yield two simultaneously nonvanishing pure charges within the original electromagnetic formulation.
Indeed, Eq.~\eqref{eq:hodge-charge-map} gives
$c{\mathsf p}=0=s{\mathsf q}$. For $cs\ne0$ both seed charges vanish;
for $s=0$ or $c=0$ at most one electric charge remains. The restriction
therefore follows from the full-field Hodge map.

Thus the Rasheed--Larsen direct--dual family combines two field profiles in its original
frame with a three-dimensional constrained charge surface. Its first-law
variations stay on that surface. The global mass result consists of the
full Rasheed regular domain, the unique root selected by the rational
scalar equation, and the two finite-area extremal limits.

\section{Common outer- and inner-horizon relations}
\label{sec:two-horizon}

\subsection{Normalization, domain, and common relations}
\label{subsec:two-horizon-common}

The four embeddings have complete outer-horizon thermodynamics on the
regular domains specified in the preceding sections, including their
finite-area extremal boundaries. Within each domain, the comparison of
two horizons applies to the states with two distinct regular horizons and
positive finite geometric entropies $S_\pm=A_\pm/4>0$.
These states form the subdomains listed in
Table~\ref{tab:two-horizon-realization}; the finite-area limits with
coincident roots lie on their boundaries. The statements below apply to
every state of the four embeddings in these specified domains.

We compare the independently defined asymptotic long-range-force
parameter $\mu_{\rm LR}$ in Eq.~\eqref{eq:mu-LR-general} with the geometric
quantities $S_\pm,T_\pm,\Omega_\pm$.
All angular velocities are measured relative to the nonrotating frame at
infinity, $\Omega_\infty=0$. We use the future-inner temperature convention
$T_-=\kappa_-^{\rm future}/(2\pi)<0$.
The first laws established in the preceding sections refer to the outer
horizon. Here the inner-horizon result concerns its entropy, signed
temperature, angular velocity, and their common quadratic combination.
An inner-horizon first law would also require its electromagnetic
potentials and magnetic patch contributions.

The quadratic horizon relation and its outer--inner counterpart are known
for the relevant STU seeds \cite{ChowCompere2014},
while their temperature and angular-velocity sign relations are given in
Ref.~\cite{CveticEtAl2018}.
Using the electromagnetic charge normalizations for each EMMD embedding, we separately perform the verification of asymptotic and geometric quantities in Section~\ref{sec:two-horizon-audit}.

On the regular outer-horizon domains, including their finite-area
extremal boundaries, the relation is
\begin{equation}
 \mu_{\rm LR}^2=(2TS)^2+\left(\frac{S\Omega_H}{\pi}\right)^2.
 \label{eq:unified-rotating-identity}
\end{equation}
For every state in the distinct-horizon subdomains, the thermal and
rotational products satisfy
\begin{equation}
 \begin{aligned}
  S_-T_-&=-S_+T_+,\\
  S_-\Omega_-&=\varepsilon\,S_+\Omega_+.
 \end{aligned}
 \label{eq:two-horizon-sign}
\end{equation}
Consequently, their common sum of squares is
\begin{equation}
 \begin{aligned}
  &(2T_-S_-)^2+\left(\frac{S_-\Omega_-}{\pi}\right)^2\\
 =&(2T_+S_+)^2+\left(\frac{S_+\Omega_+}{\pi}\right)^2
=\mu_{\rm LR}^2.
 \end{aligned}
 \label{eq:two-horizon-quadratic}
\end{equation}
The relative sign is $\varepsilon=+1$ for the equal-coupling KK,
zero-coupling Kerr--Newman, and balanced Kerr--Newman embeddings.
For the Rasheed--Larsen direct--dual embedding,
\begin{equation}
 \varepsilon=\operatorname{sgn}(J^2-A).
 \label{eq:two-horizon-RL-sign}
\end{equation}
Here $A$ is the quartic charge combination defined in
Eq.~\eqref{eq:rasheed-intrinsic-definitions}; its EMMD charge dependence
is fixed by Eq.~\eqref{eq:hodge-dressed-seed-squares}.
The surface $J^2=A$ has zero inner area and is excluded from this
two-horizon statement. The direct checks below establish both the norm
and the unsquared sign relations.

\subsection{Verification of the four embeddings}
\label{app:two-horizon-audit}
\label{sec:two-horizon-audit}

For each embedding we first calculate $\mu_{\rm LR}^2$ from its mass,
scalar charge, and electromagnetic charges, then evaluate the two
horizon products from the seed geometry and compare the results.

\subsubsection{Equal-coupling Kaluza--Klein}

The asymptotic mass, scalar charge, and electric charges give
$\mu_{\rm LR}^2=m^2$ by
Eqs.~\eqref{eq:o2-static-rotation-bridge}--\eqref{eq:o2-shift-covariant-mu}.
Independently, evaluating the seed geometry at the roots in
Eq.~\eqref{eq:o2-parameters} gives
\begin{equation}
 \begin{aligned}
  S_\pm&=2\pi m\cosh\delta(m\pm\ell),\\
  T_\pm&=\frac{\pm\ell}{4\pi m\cosh\delta(m\pm\ell)},
  &\Omega_\pm&=\frac{b}{2m\cosh\delta(m\pm\ell)}.
 \end{aligned}
 \label{eq:o2-two-horizon-packet}
\end{equation}
Hence $2T_\pm S_\pm=\pm\ell$ and $S_\pm\Omega_\pm/\pi=b$, so
\begin{equation}
  (2T_\pm S_\pm)^2+
  \left(\frac{S_\pm\Omega_\pm}{\pi}\right)^2
  =\ell^2+b^2\\
  =m^2=\mu_{\rm LR}^2.
 \label{eq:o2-two-horizon-quadratic}
\end{equation}

\subsubsection{Zero-coupling Kerr--Newman}

Here $\Sigma=0$, and the independent electric charges give
$\mu_{\rm LR}^2=M^2-X_0$ from Eq.~\eqref{eq:mu-LR-general}.
For the geometric calculation set $Y=X_0$ and use the Kerr--Newman
quantities
\begin{equation}
 \begin{aligned}
  m_Y^2&=M^2-Y,& b_Y&=J/M,\\
  \ell_Y&=\sqrt{m_Y^2-b_Y^2},&
  D_\pm&=2M^2-Y\pm2M\ell_Y .
 \end{aligned}
 \label{eq:KN-two-horizon-definitions}
\end{equation}
The two horizon roots give
\begin{equation}
 \begin{aligned}
  S_\pm&=\pi D_\pm,\\
  T_\pm&=\frac{\pm\ell_Y}{2\pi D_\pm},&
  \Omega_\pm&=\frac{b_Y}{D_\pm},
 \end{aligned}
 \label{eq:KN-two-horizon-packet}
\end{equation}
Thus $2T_\pm S_\pm=\pm\ell_Y$ and $S_\pm\Omega_\pm/\pi=b_Y$, and
\begin{equation}
 \begin{aligned}
  (2T_\pm S_\pm)^2+
  \left(\frac{S_\pm\Omega_\pm}{\pi}\right)^2
  &=\ell_Y^2+b_Y^2\\
  &=m_Y^2.
 \end{aligned}
 \label{eq:KN-two-horizon-quadratic}
\end{equation}
With $Y=X_0$, the geometric result is
$m_Y^2=M^2-X_0=\mu_{\rm LR}^2$. We write \(b_0,\ell_0\) correspond to the values of \(b_Y,\ell_Y\) evaluated at \(Y=X_0\) in Table~\ref{tab:two-horizon-realization}.

\subsubsection{Balanced Kerr--Newman}

The balanced embedding has $\Sigma=0$. Its charge map fixes the
electromagnetic term in Eq.~\eqref{eq:mu-LR-general} to
$Y=X_{\rm bal}^{\rm dy}$, so the asymptotic data give
$\mu_{\rm LR}^2=M^2-Y$.
Independently, its unchanged Kerr--Newman geometry gives
Eqs.~\eqref{eq:KN-two-horizon-definitions}--\eqref{eq:KN-two-horizon-quadratic}
with this value of $Y$. The horizon products are
$2T_\pm S_\pm=\pm\ell_Y$ and $S_\pm\Omega_\pm/\pi=b_Y$;
their norm is 
\begin{equation}
\ell_Y^2+b_Y^2=M^2-Y=\mu_{\rm LR}^2.
\end{equation}
The two Kerr--Newman embeddings share the same horizon
geometry and therefore the same geometric derivation of
the horizon relations. Their contributions to
$\mu_{\rm LR}$, however, are different because the
electromagnetic charges are embedded into the two Maxwell
fields through different charge maps.

\subsubsection{Rasheed--Larsen direct--dual embedding}

The charge norm \eqref{eq:hodge-dressed-norm} and the asymptotic identity \eqref{eq:hodge-mu} give $\mu_{\rm LR}^2=m^2$ at arbitrary scalar
boundary value. To evaluate the two horizons independently, define
\begin{equation}
 \begin{aligned}
  \alpha&=\sqrt{\mathscr P\mathscr Q},\qquad
  d=\sqrt{m^2-b^2},\\
  k&=\frac{\mathscr P\mathscr Q+4m^2}
  {2m(\mathscr P+\mathscr Q)}.
 \end{aligned}
 \label{eq:RL-two-horizon-definitions}
\end{equation}
The seed charge formulas and Eq.~\eqref{eq:rasheed-intrinsic-definitions}
give
\begin{equation}
 \mathcal F=\frac{\alpha^2m^2}{4},\qquad
 \mathcal F+A-J^2=\frac{\alpha^2k^2d^2}{4}.
 \label{eq:RL-two-horizon-radicals}
\end{equation}
The outer entropy and the signed inner combination are
\begin{equation}
 \begin{aligned}
  S_+&=\pi\alpha(m+kd),&
  \overline S_-&=\pi\alpha(m-kd),\\
  S_-&=|\overline S_-|.
 \end{aligned}
 \label{eq:RL-two-horizon-entropies}
\end{equation}
Define $\varepsilon_{\rm RL}=\sgn(m-kd)$.
Evaluating at the two radial roots and using positive $S_-$ gives
\begin{equation}
 \begin{aligned}
  T_+&=\frac{d}{2\pi\alpha(m+kd)},
  &\Omega_+&=\frac{b}{\alpha(m+kd)},\\
  T_-&=-\frac{d}{2\pi\alpha|m-kd|},
  &\Omega_-&=\frac{b}{\alpha(m-kd)}.
 \end{aligned}
 \label{eq:RL-two-horizon-responses}
\end{equation}
The absolute value keeps the inner temperature negative, while the sign
of $m-kd$ remains in the inner angular velocity. Thus
$2T_\pm S_\pm=\pm d$, $S_+\Omega_+/\pi=b$, and
$S_-\Omega_-/\pi=\varepsilon_{\rm RL}b$. Consequently,
\begin{equation}
 \begin{aligned}
  (2T_\pm S_\pm)^2+
  \left(\frac{S_\pm\Omega_\pm}{\pi}\right)^2
  &=d^2+b^2\\
  &=m^2=\mu_{\rm LR}^2.
 \end{aligned}
 \label{eq:RL-two-horizon-quadratic}
\end{equation}
The sign can be written in physical charges and angular momentum because
\begin{equation}
 J^2-A=\frac{\alpha^2}{4}(m^2-k^2d^2)
 =\frac{\alpha^2}{4}(m-kd)(m+kd),
 \label{eq:RL-two-horizon-sign-proof}
\end{equation}
with $m+kd>0$. Therefore $\varepsilon_{\rm RL}=+1$ for $J^2>A$ and
$\varepsilon_{\rm RL}=-1$ for $J^2<A$.
The metric calculation in Appendix~\ref{app:inner-geometry} explains the
absolute value in the entropy and the retained sign in the angular
velocity, and verifies regularity at the allowed inner roots.

\subsection{Regular two-horizon domains, static and extremal limits}
\label{subsec:two-horizon-domains}

For the equal-coupling KK embedding, both areas are positive on
$m>0$, $0<|b|<m$. Charged states have $\delta>0$, while $\delta=0$
gives rotating Kerr. At $b=0$ the inner area vanishes; at $|b|=m$ the
roots coincide.

For both Kerr--Newman embeddings, with $Y=X_0$ or
$Y=X_{\rm bal}^{\rm dy}$ respectively, the strict two-horizon domain is

\begin{equation}
 \begin{gathered}
  M>0,\qquad Y\ge0,\qquad J^2<M^2(M^2-Y),\\
  (Y,J)\ne(0,0).
 \end{gathered}
 \label{eq:KN-two-horizon-domain}
\end{equation}
Charged states have $Y>0$. For $Y=0$, the last condition selects
rotating Kerr, while Schwarzschild has a zero-area inner root.
Equality in the spin bound merges the horizons. This includes the
static extremal RN boundary at $J=0$, $M^2=Y>0$, which belongs to the
finite-area outer-horizon extremal family.

For Rasheed--Larsen, the zero-inner-area threshold in seed parameters is
\begin{equation}
 \begin{aligned}
 & b_*^2=m^2\frac{(\mathscr P^2-4m^2)(\mathscr Q^2-4m^2)}
  {(\mathscr P\mathscr Q+4m^2)^2},\\
 &\operatorname{sgn}(J^2-A)=\operatorname{sgn}(b^2-b_*^2).
 \end{aligned}
 \label{eq:RL-two-horizon-threshold}
\end{equation}
For ${\mathsf p}{\mathsf q}\ne0$, one has $\mathscr P,\mathscr Q>2m$ and
\begin{equation}
 \frac{b_*^2}{m^2}
 =1-\frac{4m^2(\mathscr P+\mathscr Q)^2}
 {(\mathscr P\mathscr Q+4m^2)^2},
 \qquad 0<\frac{b_*^2}{m^2}<1.
\end{equation}
Thus the zero-inner-area surface lies inside the nonextremal
outer-horizon spin range. The strict two-horizon domain is $m>0$,
$|b|<m$, $\mathscr P,\mathscr Q\ge2m$, and $J^2\ne A$.
On a one-charge axis, $A=0$, so this condition excludes the static
inner root of zero area. A static dyonic state has $J^2<A$ and finite
inner area. Appendix~\ref{app:inner-geometry} proves that the retained
inner roots have positive metric and scalar factors. Because regularity of the outer horizon alone does
not guarantee a finite-area inner horizon.

Table~\ref{tab:two-horizon-realization} collects the products and domains
in the same order as the thermodynamic branches. Seed symbols are local
to each row. In the KK row, $m,b,\ell$ are defined by
Eq.~\eqref{eq:o2-parameters}. In the balanced row, the first two lines
below use $X=X_{\rm bal}^{\rm dy}$ and abbreviate
$(m_Y,b_Y,\ell_Y)$ as $(m,b,\ell)$. The next two lines define the
zero-coupling symbols; the last line uses the Rasheed--Larsen seed
parameters from Eq.~\eqref{eq:RL-two-horizon-definitions}. Thus the
symbols $m,b$ in the KK, balanced, and Rasheed--Larsen rows refer to
their respective seeds:

\begin{equation}
 \begin{gathered}
  m^2=M^2-X,\qquad b=J/M,\\
  \ell=\sqrt{m^2-b^2}\quad\text{(balanced)},\\
  b_0=J/M,\qquad
  \ell_0=\sqrt{M^2-X_0-b_0^2}\quad\text{(zero coupling)},\\
  d=\sqrt{m^2-b^2}\quad\text{(Rasheed--Larsen)}.
 \end{gathered}
 \label{eq:two-horizon-table-definitions}
\end{equation}

\begin{table*}[t]
\caption{
Domains with two distinct regular horizons of positive finite area.
In each ordered pair, the outer-horizon quantity is listed first.
The quantities $\ell$, $\ell_0$, and $d$ are positive, and the seed
parameters have the row-specific meanings defined in the text.
In the equal-coupling KK row, $\delta>0$ gives charged states,
whereas $\delta=0$ gives nonextremal rotating Kerr because $b\ne0$.
In both Kerr--Newman rows, zero-charge states with nonzero angular
momentum are retained. The Schwarzschild limit is excluded because
it has no regular inner horizon.
The geometric horizon relations follow from the corresponding
seed solutions
\cite{Rasheed1995,Larsen2000,ChowCompere2014,CveticEtAl2018}.
}
\label{tab:two-horizon-realization}
\centering\footnotesize
\renewcommand{\arraystretch}{1.22}
\setlength{\tabcolsep}{3pt}
\newcommand{\horizoncell}[2]{\parbox[t]{#1}{\raggedright #2\par}}
\begin{tabular}{@{}llllll@{}}
\toprule
\horizoncell{0.13\textwidth}{Embedding} & \horizoncell{0.265\textwidth}{Strict two-horizon domain} & \horizoncell{0.10\textwidth}{$2T_\pm S_\pm$} & \horizoncell{0.12\textwidth}{$S_\pm\Omega_\pm/\pi$} & \horizoncell{0.10\textwidth}{$\varepsilon$} & \horizoncell{0.17\textwidth}{Excluded limits} \\
\midrule
\horizoncell{0.13\textwidth}{Equal-coupling KK} & \horizoncell{0.265\textwidth}{$m>0$, $0<|b|<m$, $\delta\ge0$} & \horizoncell{0.10\textwidth}{$(\ell,-\ell)$} & \horizoncell{0.12\textwidth}{$(b,b)$} & \horizoncell{0.10\textwidth}{$+1$} & \horizoncell{0.17\textwidth}{$b=0$: $S_-=0$; $|b|=m$: coincident roots} \\
\addlinespace[2pt]
\horizoncell{0.13\textwidth}{Zero-coupling Kerr--Newman} & \horizoncell{0.265\textwidth}{$M>0$, $X_0\ge0$, $J^2<M^2(M^2-X_0)$, $(X_0,J)\ne(0,0)$} & \horizoncell{0.10\textwidth}{$(\ell_0,-\ell_0)$} & \horizoncell{0.12\textwidth}{$(b_0,b_0)$} & \horizoncell{0.10\textwidth}{$+1$} & \horizoncell{0.17\textwidth}{$X_0=J=0$: $S_-=0$; $\ell_0=0$: extremal merger} \\
\addlinespace[2pt]
\horizoncell{0.13\textwidth}{Balanced Kerr--Newman} & \horizoncell{0.265\textwidth}{$M>0$, $X:=X_{\rm bal}^{\rm dy}\ge0$, $J^2<M^2(M^2-X)$, $(X,J)\ne(0,0)$} & \horizoncell{0.10\textwidth}{$(\ell,-\ell)$} & \horizoncell{0.12\textwidth}{$(b,b)$} & \horizoncell{0.10\textwidth}{$+1$} & \horizoncell{0.17\textwidth}{$X=J=0$: $S_-=0$; $\ell=0$: extremal merger} \\
\addlinespace[2pt]
\horizoncell{0.13\textwidth}{Rasheed--Larsen direct--dual embedding} & \horizoncell{0.265\textwidth}{$m>0$, $|b|<m$, $\mathscr P,\mathscr Q\ge2m$, $J^2\ne A$} & \horizoncell{0.10\textwidth}{$(d,-d)$} & \horizoncell{0.12\textwidth}{$(b,\varepsilon b)$} & \horizoncell{0.10\textwidth}{$\sgn(J^2-A)$} & \horizoncell{0.17\textwidth}{$J^2=A$: $S_-=0$; extremal limits: coincident roots} \\
\bottomrule
\end{tabular}
\end{table*}

The fast and slow Rasheed--Larsen extremal limits have coincident radial
roots and lie on the boundary of the distinct-horizon domain. Static
limits depend on the embedding: the equal-coupling KK inner area
vanishes, whereas nonextremal charged static Kerr--Newman and static
dyonic Rasheed--Larsen states retain positive finite inner areas. The
finite-area static extremal RN solution remains part of the complete
outer-horizon thermodynamics.

As a static normalization check, the purely electric relation of
Ref.~\cite{LuWuWuZhao2026Toda}, expressed in our scalar-charge convention,
is
\begin{equation}
 M^2+\Sigma^2-4(Q_1^2+Q_2^2)=(2TS)^2
 \quad (\phi_0=0,\ J=0).
 \label{eq:toda-static-normalization}
\end{equation}
On the balanced static loci this is the $J=0$ limit of
Eq.~\eqref{eq:unified-rotating-identity}.

\subsection{Physical meaning of the common norm and rotation signs}
\label{subsec:two-horizon-meaning}

The common relation ties an asymptotic combination of mass, scalar
charge, and electromagnetic charges to two horizon scales.
Since $T=\kappa/(2\pi)$ and $S=A_H/4$, these scales are
$2TS=\kappa A_H/(4\pi)$ and $S\Omega_H/\pi=A_H\Omega_H/(4\pi)$.
Their common norm can be read from either regular horizon of the same
black-hole state. The value can change when the state changes.
For two identical copies of a black-hole state in the same
scalar background, $\mu_{\rm LR}^2$ determines the leading,
velocity-independent force at large separation. This
monopole term depends only on the mass, scalar charge, and
electric and magnetic charges. Explicit spin couplings and
higher-multipole interactions lie beyond this leading term.
For two different black holes, the force instead depends on
products of their respective masses and charges, so their
individual values of $\mu_{\rm LR}$ are not sufficient
\cite{YangLuLu2026PRL,Heidenreich2020,CremoniniEtAl2022LongRange}.

The embeddings also show why the charge map is part of the thermodynamic
description. The same seed geometry can support different allowed
charge variations and different organizations of the electromagnetic
work term. The balanced charge plane and the Rasheed--Larsen
direct--dual charge surface provide explicit examples. The horizon norm
records their common geometric structure; the charge constraints and
potential maps determine which changes of state are allowed and what
work they require.

At a finite-area extremal boundary the outer relation reduces to
\begin{equation}
 \mu_{\rm LR}^2=\left(\frac{S\Omega_H}{\pi}\right)^2,
 \qquad T=0,\quad 0<S<\infty.
 \label{eq:two-horizon-extremal-meaning}
\end{equation}
Thus rotation prevents $\mu_{\rm LR}$ from being a universal measure of
departure from extremality. The Rasheed--Larsen fast extremal branch in
Eq.~\eqref{eq:rasheed-fast} has $\Omega_H\ne0$ and $\mu_{\rm LR}>0$.
The slow extremal branch in Eq.~\eqref{eq:rasheed-slow-1} can have
$J\ne0$ while $\Omega_H=0$ and $\mu_{\rm LR}=0$.
The common relation accommodates both extremal behaviors.

Squaring the products removes the relative rotation sign
$\varepsilon$. In the Rasheed--Larsen family that sign distinguishes
the two sides of $J^2=A$ and contains inner-horizon information absent
from the norm. Equation~\eqref{eq:two-horizon-sign} therefore remains
part of the result alongside Eq.~\eqref{eq:two-horizon-quadratic}.

\section{Conclusions}
\label{sec:conclusions}

We have obtained complete outer-horizon thermodynamics for four
exact rotating EMMD embeddings of known seeds: equal-coupling KK,
zero-coupling Kerr--Newman, balanced Kerr--Newman, and
Rasheed--Larsen direct--dual families. Within the constant-coefficient
seed/dual construction, we determine the maps that
preserve the seed metric and scalar profile and verify the target
field equations. For each family, we specify the independent
electromagnetic charges and their constraints. On the stated regular
domains, the thermodynamic quantities satisfy the first law for
allowed variations, together with the Smarr and scalar-shift
identities. The electromagnetic work includes the required magnetic
patch contributions whenever magnetic charge is present.
Single-valued mass functions on the nonextremal outer branches
follow from the Kerr--Newman mass relations, the unique physical
root of the KK cubic, and the strict monotonicity of $\mathcal F(M)$
along the physical Rasheed--Larsen scalar root.

We calculate $\mu_{\rm LR}$ from the asymptotic charges and
evaluate the horizon products separately from the geometry.
This verifies the known seed relation
\cite{ChowCompere2014,CveticEtAl2018} in the charge normalization
of each EMMD embedding:
\begin{equation*}
 \mu_{\rm LR}^2
 =(2T_\pm S_\pm)^2
 +\left(\frac{S_\pm\Omega_\pm}{\pi}\right)^2.
\end{equation*}
The comparison of both horizons applies to the domains in
Table~\ref{tab:two-horizon-realization}, where the horizons are
distinct, regular, and have positive finite area.
With $S_\pm>0$ and the future-inner temperature convention
$T_-<0$, direct calculations also give
\begin{equation*}
 T_-S_-=-T_+S_+,
 \qquad
 S_-\Omega_-=\varepsilon S_+\Omega_+.
\end{equation*}
Here $\varepsilon=+1$ for the equal-coupling KK and the two
Kerr--Newman embeddings, whereas
$\varepsilon=\operatorname{sgn}(J^2-A)$ for the Rasheed--Larsen
family. Within its nonextremal outer branch, $J^2=A$ marks
zero inner area and is excluded from the two-horizon comparison.

The outer-horizon relation remains valid at the stated regular
finite-area extremal limits. On these boundaries, the first law
is understood for variations within the chosen extremal branch.
The Rasheed--Larsen fast extremal branch has
$T_+=0$, $\Omega_+\ne0$, and $\mu_{\rm LR}>0$.
The slow extremal branch can instead have $J\ne0$ while
$T_+=\Omega_+=0$ and $\mu_{\rm LR}=0$.
Thus $\mu_{\rm LR}$ need not vanish at extremality and cannot
serve as a universal measure of departure from extremality.

The charge maps show why a common seed geometry does not fix
the thermodynamic state space of its EMMD embeddings.
The equal-coupling KK and zero-coupling Kerr--Newman families
have two independent electric charges even though their Maxwell
fields are proportional. The Rasheed--Larsen family generally
has two nonproportional fields in the original electromagnetic
frame, but its four charges obey $Q_1Q_2+P_1P_2=0$.
Its first law therefore involves three independent charge
directions on the nonzero charge surface, and its purely electric
subset has at most one nonzero electric charge.
These examples distinguish the number of independent charges
from the number of independent field profiles.
The common geometric horizon relation survives these different
charge realizations, while the allowed state variations and
electromagnetic work must be determined from the actual field
maps and gauge patches.

\begin{acknowledgments}
This work is supported by the National Natural Science Foundation of China
(NSFC) under Grant Nos.~12373022 and U1731107, the Seventh Batch of High Level
Innovative Talents of Guizhou Province under Grant No.~GCC[2023]011, and the
Special Funds for Discipline Construction and Postgraduate Education under
Grant No.~JX-2020-02.

\end{acknowledgments}

\appendix

\section{Gauge patches and horizon regularity}
\label{app:regularity}

\subsection{Gauge patches and the extremal boundary}
\label{app:gauge-patches}

The Kerr--Newman line element associated with Eq.~\eqref{eq:rho-delta} is

\begin{align}
 \dd s^2={}&-\frac{\Delta-b^2\sin^2\theta}{\rho^2}\dd t^2
 -\frac{2b(2Mr-\mathcal Z^2)\sin^2\theta}{\rho^2}
 \dd t\,\dd\varphi\notag\\
 &+\frac{\rho^2}{\Delta}\dd r^2+\rho^2\dd\theta^2\notag\\
 &+\frac{\sin^2\theta}{\rho^2}
 \left[(r^2+b^2)^2-b^2\Delta\sin^2\theta\right]
 \dd\varphi^2.
 \label{eq:KN-metric}
\end{align}
Introduce ingoing coordinates

\begin{equation}
 \dd v=\dd t+\frac{r^2+b^2}{\Delta}\dd r,
 \qquad
 \dd\widetilde\varphi=\dd\varphi+\frac{b}{\Delta}\dd r.
 \label{eq:ingoing}
\end{equation}
After substitution into Eq.~\eqref{eq:KN-metric}, all inverse powers of
$\Delta$ cancel. The radial terms produced in the gauge potential are
locally exact and can be removed separately in the two patches. This gives

\begin{align}
 A_{\rm in}^{(N/S)}={}&-\frac{4\mathfrak q r}{\rho^2}
 \left(\dd v-b\sin^2\theta\,\dd\widetilde\varphi\right)\notag\\
 &+\frac{4\mathfrak p\cos\theta}{\rho^2}
 \left[b\,\dd v-(r^2+b^2)\dd\widetilde\varphi\right]
 \mathbin{\pm}4\mathfrak p\,\dd\widetilde\varphi .
 \label{eq:ingoing-potential}
\end{align}
Each potential is smooth at $\Delta=0$ and on its own rotation axis.
Their difference is
$A_{\rm in}^{(N)}-A_{\rm in}^{(S)}=8\mathfrak p\,\dd\widetilde\varphi$.
The target rescalings in Eq.~\eqref{eq:balanced-difference} give
$A_i^{(N)}-A_i^{(S)}=8P_i\,\dd\widetilde\varphi$.

The full finite-area extremal boundary is specified directly by
$M^2=b^2+\mathcal Z^2>0$. This includes rotating extremal black holes and
static extremal RN with $b=0$, $M=\mathcal Z>0$. On this boundary,

\begin{equation}
 \Delta=(r-M)^2,
 \qquad
 \rho_H^2=M^2+b^2\cos^2\theta>0.
\end{equation}

Since $r_H=M>0$, the horizon is separated from the Kerr ring singularity
at $r=0$, $\theta=\pi/2$. Its area is $4\pi(M^2+b^2)$, reducing to
$4\pi\mathcal Z^2$ for static extremal RN.
The axes are regular for azimuthal period $2\pi$, since

\begin{equation}
 \lim_{\theta\to0}
 \frac{g_{\varphi\varphi}}{\theta^2g_{\theta\theta}}=1,
 \qquad
 \lim_{\theta\to\pi}
 \frac{g_{\varphi\varphi}}{(\pi-\theta)^2g_{\theta\theta}}=1,
\end{equation}
and $g_{t\varphi}=O(\sin^2\theta)$.
The regular potentials satisfy
$A_{i\widetilde\varphi}^{(N)}=O(\sin^2\theta)$ near the north axis and
$A_{i\widetilde\varphi}^{(S)}=O(\sin^2\theta)$ near the south axis.
The gauge fields are therefore regular on their respective patches.
The limiting horizon is smooth, connected, and has spatial topology $S^2$.

\subsection{Seed geometry at the inner root}
\label{app:inner-geometry}

The entropies in Eq.~\eqref{eq:RL-two-horizon-entropies} can be directly verified using the geometry of the seed solution. In Larsen's notation\cite{Larsen2000}, with the auxiliary parameters renamed as
$(\mathscr P,\mathscr Q,b)$, the four-dimensional metric is
\begin{equation}
 \begin{aligned}
 \dd s^2={}&-\frac{H_3}{W}(\dd t+B_\varphi\dd\varphi)^2\\
 &+W\left(\frac{\dd r^2}{\Delta}+\dd\theta^2
 +\frac{\Delta}{H_3}\sin^2\theta\,\dd\varphi^2\right),
 \end{aligned}
\end{equation}
where $W=\sqrt{H_1H_2}$ and $H_3=\Delta-b^2\sin^2\theta$.
Set $\Lambda_\pm=\alpha(m\pm kd)$, with $\alpha,k,d$ as in
Eq.~\eqref{eq:RL-two-horizon-definitions}.
For $b\ne0$, the seed shift at the roots obeys
\begin{equation}
 \begin{gathered}
 B_\varphi(r_\pm)=-\frac{\Lambda_\pm}{b},\qquad
 \Omega_\pm=\frac{b}{\Lambda_\pm},\\
 \left.\sqrt{g_{\theta\theta}g_{\varphi\varphi}}\right|_{r_\pm}
 =|\Lambda_\pm|\sin\theta.
 \end{gathered}
\end{equation}
Thus $S_\pm=\pi|\Lambda_\pm|$. The area uses an absolute value, while
the denominator of $\Omega_-$ retains its sign.

The metric and scalar factors are positive at the retained inner root.
To see this, write
\begin{equation}
 \begin{aligned}
& u=\frac{\mathscr P}{2m}\ge1,\qquad
 v=\frac{\mathscr Q}{2m}\ge1,\qquad
 \beta=\frac{|b|}{m},\qquad \tau=\frac d m,\\
&\beta^2+\tau^2=1,\qquad
 k=\frac{uv+1}{u+v}\ge1,\qquad \zeta=\sqrt{k^2-1}.
 \end{aligned}
\end{equation}
Choose the polar orientation so that the linear term below is negative.
At $r_-=m-d$ and $z=\cos\theta$, the source functions give
\begin{equation*}
 \frac{H_1(r_-,z)}{m^2}
 =\beta^2z^2-2u\zeta\beta z+2u(k-\tau)-\beta^2.
\end{equation*}
For $H_2$, replace $u$ by $v$ and $z$ by $-z$.
If $\beta=0$, the retained states have both charges nonzero and $k>1$,
so both functions are positive. For $\beta>0$ and $u\zeta\ge\beta$,
the minimum of $H_1$ on $[-1,1]$ occurs at $z=1$. It is positive because
$\tau+\zeta\beta\le k$, with equality only at $\tau=1/k$,
the excluded surface $m=kd$.
If $u\zeta<\beta$, the minimum occurs at the vertex and is
\begin{equation}
 [1-\tau-u(k-1)]\,[u(k+1)-1-\tau]>0.
\end{equation}
The second factor is positive, and the first is positive since
\begin{equation}
 u(k-1)=\frac{u\zeta^2}{k+1}
 <\frac{\beta^2}{u(k+1)}
 \le\frac{\beta^2}{1+\tau}=1-\tau.
\end{equation}
The same argument applies to $H_2$. Moreover,
$\partial_rH_1=2r+\mathscr P-2m$ and
$\partial_rH_2=2r+\mathscr Q-2m$ are nonnegative for $r\ge r_-$.
Their positivity therefore extends from the inner root to the exterior.
At either axis,
\begin{equation}
 (H_1H_2)|_{r_-,\,z=\pm1}
 =\mathscr P\mathscr Q(m-kd)^2=\Lambda_-^2.
\end{equation}
This also gives the regular polar area element for azimuthal period $2\pi$.
Together with a simple radial root and horizon-penetrating seed coordinates,
these properties give a regular local horizon with a smooth scalar ratio.
At $b=0$, the retained dyonic solution has $H_1,H_2>0$ and the static
horizon follows directly from the same metric. The target maps preserve
this geometry and the regular seed fields, which supplies the local
regularity input for Table~\ref{tab:two-horizon-realization}.

\section{Charge-space geometry and global mass inversion}
\label{app:charge-space}

This appendix gives the detailed charge-space and mass-reconstruction
checks used in Sec.~\ref{sec:hodge}. We first show that the nonzero
charge configurations form a smooth three-dimensional surface and
that each configuration has two equivalent seed descriptions.
We then check the electromagnetic work, including the magnetic
patch data, and prove the uniqueness of the mass and scalar root.
The final subsection checks the reconstructed seed domain and
its boundary limits. These results apply to the family and physical
domains specified in the main text.

\subsection{Nonzero charge surface and equivalent seed parameters}
\label{app:hodge-charge-geometry}

Order the target charges as
\begin{equation}
 \mathbf C=(Q_1,P_1,Q_2,P_2).
 \label{eq:target-charges}
\end{equation}
At fixed scalar boundary value, the seed parameters are
$(\mathsf q,\mathsf p,\vartheta)$, with the two seed charges
not both zero and the angle defined modulo $2\pi$.
Their parameter space is
\begin{equation}
 S^1=\mathbb R/(2\pi\mathbb Z),\quad
 \mathcal D^\times=(\mathbb R^2\setminus\{(0,0)\})\times S^1.
 \label{eq:seed-parameters}
\end{equation}
With $c=\cos\vartheta$ and $s=\sin\vartheta$, the charge map is
\begin{equation}
 \Pi(\mathsf q,\mathsf p,\vartheta)
 =(c\mathsf q,c\mathsf p,-s\mathsf p,s\mathsf q).
 \label{eq:hodge-charge-map}
\end{equation}
Define $\mathcal C_{\rm ch}(\mathbf C)=Q_1Q_2+P_1P_2$.
For $\mathbf C=\Pi(\mathsf q,\mathsf p,\vartheta)$,
\begin{equation}
 \mathcal C_{\rm ch}(\mathbf C)=0,\quad
 \|\mathbf C\|^2=\mathsf q^2+\mathsf p^2>0.
 \label{eq:hodge-charge-map-identities}
\end{equation}
Thus the image lies in the nonzero constrained charge set
\begin{equation}
 \mathcal N^\times
 :=\{\mathbf C\in\mathbb R^4\setminus\{0\}:
 \mathcal C_{\rm ch}(\mathbf C)=0\}.
 \label{eq:hodge-charge-surface}
\end{equation}

Differentiating the map gives the three charge variations
\begin{equation}
 \begin{aligned}
 &\partial_{\mathsf q}\mathbf C=(c,0,0,s),\quad
 \partial_{\mathsf p}\mathbf C=(0,c,-s,0),\\
 &\partial_\vartheta\mathbf C
 =(-s\mathsf q,-s\mathsf p,-c\mathsf p,c\mathsf q).
 \end{aligned}
 \label{eq:hodge-charge-tangents}
\end{equation}
Using the ordinary Euclidean inner product on the four charge
components, these vectors are mutually orthogonal. Their scalar
products form the matrix
\begin{equation}
 G_{\mathbf C}=\operatorname{diag}(1,1,\mathsf q^2+\mathsf p^2).
 \label{eq:hodge-charge-gram}
\end{equation}
All three diagonal entries are positive on $\mathcal D^\times$,
so the charge map has rank three. In particular, changing the
angle gives a charge variation independent of those obtained
by changing either seed charge.

The constraint gradient is
\begin{equation}
 \nabla\mathcal C_{\rm ch}=(Q_2,P_2,Q_1,P_1),\quad
 \|\nabla\mathcal C_{\rm ch}\|^2=\|\mathbf C\|^2.
 \label{eq:hodge-constraint-gradient}
\end{equation}
It is nonzero everywhere on $\mathcal N^\times$.
The constraint can therefore be solved locally for one charge
in terms of the other three, so $\mathcal N^\times$ is smooth
and three-dimensional. Together with the rank calculation,
this shows that the seed charges and distribution angle give
local coordinates after choosing one seed description.

To recover seed parameters from any nonzero target configuration,
arrange the charges in the matrix
\begin{equation}
 B:=\begin{pmatrix}Q_1&P_2\\P_1&-Q_2\end{pmatrix}
 =\begin{pmatrix}\mathsf q\\\mathsf p\end{pmatrix}
 \begin{pmatrix}c&s\end{pmatrix},\quad
 \det B=-\mathcal C_{\rm ch}.
 \label{eq:hodge-rank-one-factorization}
\end{equation}
A nonzero matrix with zero determinant has rank one. It can
therefore be written as a nonzero column times a nonzero row.
Normalize the row to unit length and absorb its original length
into the column. The row becomes $(\cos\vartheta,\sin\vartheta)$,
and the column gives $(\mathsf q,\mathsf p)^{\mathsf T}$.
This proves that every point of $\mathcal N^\times$ is reached.
The normalized factorization is unique up to reversing both
factors. Thus the only two seed descriptions are those related
by Eq.~\eqref{eq:hodge-double-cover}.

These choices vary smoothly near each nonzero target configuration.
Choose a row of $B$ that is nonzero at the point under consideration.
It remains nonzero in a sufficiently small neighborhood.
Normalizing it gives a smooth local choice of $(c,s)$, and
\begin{equation}
 \begin{pmatrix}\mathsf q\\\mathsf p\end{pmatrix}
 =B\begin{pmatrix}c\\s\end{pmatrix}.
 \label{eq:hodge-local-seed-reconstruction}
\end{equation}
Reversing both factors gives the second choice. Each sufficiently
small target neighborhood therefore has two separate smooth
seed descriptions. This is the two-sheeted covering property.
Identifying the two descriptions gives
\begin{equation}
 \mathcal D^\times/\mathbb Z_2\simeq\mathcal N^\times.
 \label{eq:hodge-charge-quotient}
\end{equation}
The identification removes a discrete duplication without
reducing the number of continuous charge parameters.

If the map is extended to $\mathsf q=\mathsf p=0$, every angle
gives $\mathbf C=0$. The angle derivative then vanishes, whereas
the two charge derivatives remain independent, so the extended
map has rank two there. The zero-charge point is excluded from
the smooth surface considered above. Individual nonzero charge
axes remain regular points because the constraint gradient is
nonzero there. The two-description result also remains valid
on these axes. Black-hole regularity imposes the additional
conditions stated in the main text.

\subsection{Electromagnetic work and magnetic gauge patches}
\label{app:hodge-work-proof}

Substitute the charge map and the potential map
\eqref{eq:hodge-potential-map} into the electromagnetic work.
For all variations generated by the seed charges and angle,
\begin{equation}
 \begin{aligned}
 \sum_i(\Phi_i\dd Q_i+\Psi_i\dd P_i)
 &=(c^2+s^2)(\Phi\dd\mathsf q+\Psi\dd\mathsf p)\\
 &\quad+(c\dd c+s\dd s)(\mathsf q\Phi+\mathsf p\Psi).
 \end{aligned}
 \label{eq:hodge-work-expansion-app}
\end{equation}
Since $c^2+s^2=1$ and $c\dd c+s\dd s=0$, this is precisely
Eq.~\eqref{eq:hodge-work-map}. All angle terms cancel.
Reversing both seed charges reverses both seed potentials in
the chosen reference gauge. Together with the sign change of
$c,s$, this leaves the target potentials unchanged. The two
seed descriptions therefore give the same work on the target
charge surface.

The seed work includes the magnetic overlap and axis-boundary
contributions. To track their mapping, let $\mathcal A$ and
$\widetilde{\mathcal A}$ be local seed gauge potentials, with
\begin{equation}
 \dd\mathcal A=H,\qquad
 \dd\widetilde{\mathcal A}=-\e^{a\phi}\star H.
 \label{eq:hodge-seed-local-potentials}
\end{equation}
The target potentials can be chosen locally as
\begin{equation}
 \begin{aligned}
 A_1&=c\mathcal A,&\widetilde A_1&=c\widetilde{\mathcal A},\\
 A_2&=-s\widetilde{\mathcal A},&\widetilde A_2&=s\mathcal A.
 \end{aligned}
 \label{eq:hodge-local-potential-map}
\end{equation}
The coefficients are constant on spacetime.
If $\lambda$ and $\widetilde\lambda$ are the seed north--south
transition functions, the target transition functions are
\begin{equation}
 \begin{aligned}
 \lambda_1&=c\lambda,&\widetilde\lambda_1&=c\widetilde\lambda,\\
 \lambda_2&=-s\widetilde\lambda,&\widetilde\lambda_2&=s\lambda.
 \end{aligned}
 \label{eq:hodge-patch-transition-map}
\end{equation}
Thus the same field map carries the gauge-patch data used in
the seed first law. Equation~\eqref{eq:hodge-work-map} uses
the complete seed and target work, including these patch
contributions, rather than a potential defined in only one
patch. This step relies on the seed work being defined with
the required magnetic patch prescription.

The work is evaluated only for variations preserving the charge
constraint. Its differential is
\begin{equation}
 \dd\mathcal C_{\rm ch}
 =Q_2\dd Q_1+P_2\dd P_1+Q_1\dd Q_2+P_1\dd P_2.
 \label{eq:hodge-charge-constraint-covector}
\end{equation}
It vanishes on every allowed tangent variation, but not on an
arbitrary variation in the four-charge coordinate space.
For any smooth function $f$, consider changing the four work
coefficients by
\begin{equation}
 \begin{aligned}
 \Phi_1'&=\Phi_1+fQ_2,&\Psi_1'&=\Psi_1+fP_2,\\
 \Phi_2'&=\Phi_2+fQ_1,&\Psi_2'&=\Psi_2+fP_1.
 \end{aligned}
 \label{eq:hodge-normal-potential-ambiguity}
\end{equation}
Here $f$ is unrelated to the gauge transition functions above.
Writing $\Delta_f$ for the change induced by this replacement,
we find
\begin{equation}
 \Delta_f\!\left[\sum_i(\Phi_i\dd Q_i+\Psi_i\dd P_i)\right]
 =f\dd\mathcal C_{\rm ch}=0,
 \label{eq:hodge-normal-work-invariance}
\end{equation}
\begin{equation}
 \Delta_f\!\left[\sum_i(\Phi_iQ_i+\Psi_iP_i)\right]
 =2f\mathcal C_{\rm ch}=0,
 \label{eq:hodge-normal-Smarr-invariance}
\end{equation}
\begin{equation}
 \begin{aligned}
 \Delta_f\!\left[\sum_i a_i(Q_i\Phi_i-P_i\Psi_i)\right]
 &=f(a_1+a_2)\\
 &\quad\times(Q_1Q_2-P_1P_2)=0.
 \end{aligned}
 \label{eq:hodge-normal-invariances}
\end{equation}
The first equality uses the allowed-variation condition, the
second uses the charge constraint, and the third uses $a_1+a_2=0$.
The fields, asymptotic reference, and magnetic gauge patches
select the potentials in Eq.~\eqref{eq:hodge-potential-map}.
The replacement above changes only how the work coefficients
are extended to variations that leave the charge surface.
It is not a Maxwell gauge transformation.

\subsection{Unique scalar root and mass function}
\label{app:rl-mass-proof}

We now establish the mass reconstruction stated in
Sec.~\ref{subsec:rl-unique-mass}. All seed charges in this
subsection use $\phi_0=0$. For an arbitrary scalar boundary
value, the same argument applies after replacing them with
the weighted charges in
Eq.~\eqref{eq:hodge-dressed-seed-squares}.

We first determine the physical scalar charge at fixed real
seed charges $(\mathsf q,\mathsf p)$ and mass $M>M_0$.
We then show that $\mathcal F(M)$ increases smoothly from
zero to infinity along this scalar-root branch. This
monotonicity, together with the outer-horizon conditions
 in Appendix~\ref{app:rl-reconstruction}, determines
the mass uniquely from the entropy, angular momentum,
and charges.

Move all terms of Eq.~\eqref{eq:rasheed-rational} to the
left and define
\begin{equation}
 \mathcal R(\Sigma)
 =\frac{4\mathsf q^2}{\Sigma+\sqrt3M}
 +\frac{4\mathsf p^2}{\Sigma-\sqrt3M}
 -\frac{2\Sigma}{3}.
 \label{eq:rasheed-root-function}
\end{equation}
The physical scalar charge is a zero of $\mathcal R$ in
the interval $-\sqrt3M<\Sigma<\sqrt3M$. At fixed mass and
charges,
\begin{equation}
 \mathcal R'(\Sigma)
 =-\frac{4\mathsf q^2}{(\Sigma+\sqrt3M)^2}
 -\frac{4\mathsf p^2}{(\Sigma-\sqrt3M)^2}
 -\frac23<0.
 \label{eq:rasheed-root-monotonicity}
\end{equation}
Thus $\mathcal R$ decreases strictly throughout the
physical interval.

When both seed charges are nonzero, $\mathcal R$ tends
to $+\infty$ as $\Sigma\to(-\sqrt3M)^+$ and to $-\infty$
as $\Sigma\to(\sqrt3M)^-$. Continuity and strict
monotonicity therefore give exactly one physical root.

When exactly one seed charge vanishes, one endpoint
limit remains finite:
\begin{equation}
 \begin{aligned}
 \lim_{\Sigma\to(-\sqrt3M)^+}\mathcal R(\Sigma)
 &=\frac{2(M^2-\mathsf p^2)}{\sqrt3M}>0
 &&(\mathsf q=0),\\
 \lim_{\Sigma\to(\sqrt3M)^-}\mathcal R(\Sigma)
 &=\frac{2(\mathsf q^2-M^2)}{\sqrt3M}<0
 &&(\mathsf p=0).
 \end{aligned}
 \label{eq:rasheed-axis-root-limits}
\end{equation}
These signs follow from $M>M_0$: on the electric axis
$M_0=|\mathsf q|$, while on the magnetic axis
$M_0=|\mathsf p|$. The other endpoint diverges with the
opposite sign, so the physical root is unique on each
nonzero one-charge axis as well. When both charges
vanish, $\mathcal R=-2\Sigma/3$ and the unique root is
$\Sigma=0$.

Since $\mathcal R'$ never vanishes, the implicit function
theorem ensures that the physical scalar root varies
smoothly with $M$ for $M>M_0$. Write this root as
$\Sigma=\sqrt3\sigma(M)$, with $-M<\sigma(M)<M$, and set
\begin{equation}
 x=M+\sigma>0,\qquad y=M-\sigma>0.
 \label{eq:rasheed-positive-variables}
\end{equation}
The scalar equation and the definition of $V$ give
\begin{equation}
 \mathsf q^2=\frac{x(x^2-V)}{4(x+y)},\qquad
 \mathsf p^2=\frac{y(y^2-V)}{4(x+y)}.
 \label{eq:rasheed-charge-factorization}
\end{equation}
These expressions relate the sign of $V$ to the
charge-dependent mass threshold $M_0$.

For $V>0$, they imply
\begin{equation}
 \begin{aligned}
 |\mathsf q|^{2/3}
 &<\frac{x}{[4(x+y)]^{1/3}},\qquad
 |\mathsf p|^{2/3}
 <\frac{y}{[4(x+y)]^{1/3}},\\
 |\mathsf q|^{2/3}+|\mathsf p|^{2/3}
 &<\frac{x+y}{[4(x+y)]^{1/3}}
 =M^{2/3}.
 \end{aligned}
 \label{eq:rasheed-mass-bound-inequality}
\end{equation}
For $V<0$, both individual inequalities and their sum
reverse. For $V=0$ with $x,y>0$, they become equalities.
Using
$M_0=(|\mathsf q|^{2/3}+|\mathsf p|^{2/3})^{3/2}$,
we therefore obtain
\begin{equation}
 V>0\quad\Longleftrightarrow\quad M>M_0
 \label{eq:rasheed-mass-bound}
\end{equation}
on the scalar-root branch with $x,y>0$.
The one-charge and neutral endpoints require separate
limits because $x$ or $y$ can vanish there. 

We next determine how $\mathcal F$ changes with $M$
while holding the seed charges fixed. Rewrite the scalar
equation as
\begin{equation}
 G(M,\sigma)
 =\sigma-\frac{2\mathsf q^2}{M+\sigma}
 +\frac{2\mathsf p^2}{M-\sigma}=0.
 \label{eq:rasheed-rescaled-scalar-equation}
\end{equation}
Its derivative with respect to $\sigma$ at fixed $M$
is positive:
\begin{equation}
 G_\sigma
 =1+\frac{2\mathsf q^2}{(M+\sigma)^2}
 +\frac{2\mathsf p^2}{(M-\sigma)^2}>0.
 \label{eq:rasheed-G-positive}
\end{equation}
Using Eq.~\eqref{eq:rasheed-charge-factorization} on the
physical root gives
\begin{equation}
 G_\sigma=\frac{3U-V}{2U},\qquad
 \frac{\dd\sigma}{\dd M}
 =-\frac{\sigma(U+V)}{M(3U-V)}.
 \label{eq:rasheed-sigma-derivative}
\end{equation}
Differentiating $\mathcal F=UV$ along this root then yields
\begin{equation}
 \frac{\dd\mathcal F}{\dd M}
 =\frac{2U(U+V)}{M}>0.
 \label{eq:rasheed-F-monotonicity}
\end{equation}
Here $M>0$, $U=xy>0$, and $V>0$. Hence $\mathcal F(M)$
is strictly increasing throughout the physical
nonextremal branch.

The endpoint analysis in
Appendix~\ref{app:rl-reconstruction} gives
$\mathcal F\to0$ as $M\downarrow M_0$.
At large $M$, the rational scalar equation gives
$\sigma=O(M^{-1})$, so
$\mathcal F\sim M^4\to\infty$.
Together with the positive derivative, these limits
show that $\mathcal F(M)$ maps $(M_0,\infty)$ smoothly
and one-to-one onto $(0,\infty)$.

Thus every positive value of $\mathcal F$ selects
exactly one mass $M>M_0$ at fixed seed charges.
Appendix~\ref{app:rl-reconstruction} completes the
reconstruction by showing that the entropy and angular
momentum select a regular nonextremal outer horizon
precisely when Eq.~\eqref{eq:rasheed-intrinsic-domain}
holds. For these data,
Eq.~\eqref{eq:rasheed-intrinsic-mass} therefore determines
a unique physical mass. The argument includes the
nonextremal one-charge axes, and the neutral limit
reproduces Kerr.

\subsection{Seed reconstruction and boundary limits}
\label{app:rl-reconstruction}

Equation~\eqref{eq:rasheed-charge-factorization}, together with
$x,y>0$ and $V>0$, gives
\begin{equation}
 x^2\ge V,\qquad y^2\ge V.
 \label{eq:rasheed-reconstruction-inequalities}
\end{equation}
Equality holds in the first relation on $\mathsf q=0$ and in
the second on $\mathsf p=0$.
With $m=\sqrt V$, $\mathscr Q=2x$, and $\mathscr P=2y$,
these inequalities give $\mathscr P,\mathscr Q\ge2m$.
The reconstructed parameters reproduce the seed charge formulas.

The same factorization gives
\begin{equation}
 x^2-4\mathsf q^2=\frac{x(U+V)}{x+y},\qquad
 y^2-4\mathsf p^2=\frac{y(U+V)}{x+y}.
 \label{eq:rasheed-positive-charge-factors}
\end{equation}
Both factors are positive. Their product satisfies
\begin{equation}
 \mathcal F+A=(x^2-4\mathsf q^2)(y^2-4\mathsf p^2)
 =\frac{U(U+V)^2}{4M^2}.
 \label{eq:rasheed-charge-product}
\end{equation}
This is the denominator used in the reconstruction of $b$.
In seed parameters,
\begin{equation}
 \sqrt{\mathcal F+A}
 =\frac{\sqrt{\mathscr P\mathscr Q}
 (\mathscr P\mathscr Q+4m^2)}{4(\mathscr P+\mathscr Q)}.
 \label{eq:rasheed-angular-reconstruction}
\end{equation}
The seed angular momentum is therefore
$J=b\sqrt{\mathcal F+A}/m$, and the signed rotation parameter is
\begin{equation}
 b=\frac{J\sqrt V}{\sqrt{\mathcal F+A}}.
 \label{eq:rasheed-signed-b}
\end{equation}
In particular, its sign is fixed by $J$.

For a regular nonextremal seed, $d=\sqrt{m^2-b^2}>0$, and
\begin{equation}
 \begin{aligned}
 \sqrt{\mathcal F}&=\frac{m\sqrt{\mathscr P\mathscr Q}}{2},\\
 \sqrt{\mathcal F+A-J^2}
 &=\frac{\sqrt{\mathscr P\mathscr Q}
 (\mathscr P\mathscr Q+4m^2)d}{4m(\mathscr P+\mathscr Q)}.
 \end{aligned}
 \label{eq:rasheed-seed-entropy-radicals}
\end{equation}
Adding these expressions reproduces $S/(2\pi)$ from the seed
entropy formula. This verifies Eq.~\eqref{eq:rasheed-intrinsic-entropy}
directly, as well as its agreement with the STU entropy expression
cited in the main text.

Solving the entropy equation for its two square roots gives
\begin{equation}
 \sqrt{\mathcal F}=\frac{L^2+J^2-A}{2L},\quad
 \sqrt{\mathcal F+A-J^2}=\frac{L^2-J^2+A}{2L}.
 \label{eq:rasheed-radical-inversion}
\end{equation}
For $L>0$, both expressions are positive exactly when
$L^2>|J^2-A|$. Squaring the first gives
Eq.~\eqref{eq:rasheed-intrinsic-mass}. Its right-hand side is
positive on this domain, so the strictly increasing function
$\mathcal F(M)$ selects exactly one $M>M_0$.
For this mass, the reconstructed rotation satisfies
\begin{equation}
 m^2-b^2
 =\frac{V(\mathcal F+A-J^2)}{\mathcal F+A}>0.
 \label{eq:rasheed-reconstructed-horizon-gap}
\end{equation}
The reconstruction therefore gives $|b|<m$, as well as
$\mathscr P,\mathscr Q\ge2m$. Conversely, every regular seed
in this nonextremal domain has two positive entropy square roots
and satisfies the stated inequality. This proves that the mass
equation covers precisely the specified nonextremal outer branch.

We conclude by examining the lower-mass endpoints employed in the monotonicity argument. If both seed charges are nonzero,
the rational equation still has a unique interior root at $M=M_0$.
The charge factorization then gives $V=0$, while $U$ remains
finite and positive. Hence $\mathcal F=UV\to0$ along the open
physical branch.

On a one-charge axis, the physical scalar roots can be written
explicitly as
\begin{equation}
 \begin{aligned}
 \sigma&=\frac{-M+\sqrt{M^2+8\mathsf q^2}}{2}
 &&(\mathsf p=0,\ \mathsf q\ne0),\\
 \sigma&=\frac{M-\sqrt{M^2+8\mathsf p^2}}{2}
 &&(\mathsf q=0,\ \mathsf p\ne0).
 \end{aligned}
 \label{eq:rasheed-one-charge-scalar-roots}
\end{equation}
As $M\downarrow M_0=|\mathsf q|$ on the electric axis,
$\sigma\to M_0$ and $y\to0$; the charge factorization gives
$V=y^2\to0$. On the magnetic axis,
$M_0=|\mathsf p|$, $\sigma\to-M_0$, and $x\to0$, with
$V=x^2\to0$. In both cases $\mathcal F\to0$.
The endpoint rational expression contains a $0/0$ term and must
be interpreted as this limit of the physical branch, not as
an additional root of the denominator-cleared cubic.

At zero seed charge, the solution is
\begin{equation}
 \sigma=0,\quad M_0=0,\quad U=V=M^2,\quad\mathcal F=M^4.
 \label{eq:rasheed-neutral-mass-limit}
\end{equation}
The mass equation reduces to the Kerr relation
\begin{equation}
 M^2=\frac{L^2+J^2}{2L},\qquad L>|J|.
 \label{eq:rasheed-neutral-Kerr-mass}
\end{equation}
In particular, $\mathcal F\to0$ as $M\downarrow0$.
Together, these limits complete the proof that $\mathcal F(M)$
maps $(M_0,\infty)$ onto $(0,\infty)$ for every fixed seed
charge pair.

Finally, the lower endpoint $M=M_0$ must not be identified
with every extremal limit. The two finite-area boundaries
of the entropy relation are
\begin{equation}
 \begin{aligned}
 L^2&=J^2-A>0,\qquad \mathcal F=L^2>0
 &&\text{(fast)},\\
 L^2&=A-J^2>0,\qquad \mathcal F=0
 &&\text{(slow)}.
 \end{aligned}
 \label{eq:rasheed-mass-boundary-types}
\end{equation}
On the fast boundary, $|b|=m>0$ and $M>M_0$.
On the slow boundary, both seed charges are nonzero,
$M\downarrow M_0$, $m\to0$, and $b/m\to J/\sqrt A$ with
$|J|<\sqrt A$. The closures of the two finite-area branches meet at
$L=0$ and $J^2=A$, where the horizon area vanishes; this
zero-area locus is excluded. These limits do not enlarge the nonextremal
domain proved above.

\clearpage
\bibliography{refs_verified}

\end{document}